\documentclass[11pt]{article}
\usepackage{amsmath,amssymb,amsthm} 
\usepackage{bbm}
\usepackage{stmaryrd}
\usepackage{graphicx}
\usepackage{paralist}
\setdefaultenum{(i)}{}{}{}
\usepackage[margin=1in]{geometry}
\usepackage[round]{natbib}
\defcitealias{gerhold.al.11}{GGMKS}

\newcommand{\nada}[1]{}

\newcommand\ve{\varepsilon}

\DeclareMathAlphabet\scr{U}{scr}{m}{n}
\SetMathAlphabet\scr{bold}{U}{scr}{b}{n}
  \DeclareFontFamily{U}{scr}{\skewchar\font'177}%
  \DeclareFontShape{U}{scr}{m}{n}{<-6>rsfs5<6-8>rsfs7<8->rsfs10}{}%
  \DeclareFontShape{U}{scr}{b}{n}{<-6>rsfs5<6-8>rsfs7<8->rsfs10}{}%

\newtheorem{theorem}{Theorem}[section]

\newtheorem{definition}[theorem]{Definition}
\newtheorem{lemma}[theorem]{Lemma}
\newtheorem{proposition}[theorem]{Proposition}

\theoremstyle{definition}

\newcommand{\tilmu}{{\tilde\mu}}

\newcommand{\tilsigma}{{\tilde\sigma}}

\newcommand{\esp}[2][E]{#1\left[#2\right]}
\newcommand{\cale}{\mathcal E}

\newcommand{\eps}{\varepsilon}

\numberwithin{equation}{section}

\renewenvironment{thebibliography}[1]{%
\begin{oldthebibliography}{#1}%
\setlength{\baselineskip}{.9em}
\linespread{.9}
\small
\setlength{\parskip}{0ex}%
\setlength{\itemsep}{.1em}%
}%
{%
\end{oldthebibliography}%
}

\begin{document}

\title{Portfolio Choice with Transaction Costs: a User's Guide
\footnote{The authors are grateful to Ren Liu for proofreading the manuscript.}
}
\author{Paolo Guasoni\thanks{Boston University, Department of Mathematics and Statistics, 111 Cummington Street Boston, MA 02215, USA.
Dublin City University, School of Mathematical Sciences, Glasnevin, Dublin 9, Ireland, email: \texttt{guasoni@bu.edu}. Partially supported by the ERC (278295), NSF (DMS-0807994, DMS-1109047), SFI (07/MI/008, 07/SK/M1189, 08/SRC/FMC1389), and FP7 (RG-248896).}
\and
Johannes Muhle-Karbe\thanks{ETH Z\"urich, Departement Mathematik, R\"amistrasse 101, CH-8092, Z\"urich, Switzerland, and Swiss Finance Institute, email:
\texttt{johannes.muhle-karbe@math.ethz.ch}. Partially supported by the National Centre of Competence in Research ``Financial Valuation and Risk Management'' (NCCR FINRISK), Project D1 (Mathematical Methods in Financial Risk Management), of the Swiss National Science Foundation (SNF).}}

\maketitle

\begin{abstract}
Recent progress in portfolio choice has made a wide class of problems involving transaction costs tractable. We review the basic approach to these problems, and outline some directions for future research.
\end{abstract}




\newpage
\section{Introduction}

Transaction costs, originally considered one of many imperfections that are best neglected, have now become a very active and fast-growing theme in Mathematical Finance.\footnote{The Mathscinet database shows only nine publications with  ``transaction costs'' in their title in the eighties (1980-1989). This figure rises to 52 in the nineties (1990-1999), and to 278 in the naughties (2000-2009).} From the outset, such a growth may seem puzzling, since over the same period transaction costs  have dramatically declined across financial markets, as stock exchanges have been fully automated, and paper trades replaced by electronic settlements. 
In fact, the interest for transaction costs reflects both the increased attention to robustness of financial models, and the growing role of high-frequency trading. The decline of bid-ask spreads has sparked a huge increase in trading volume, and high-volume strategies require a careful understanding of the effects of frictions on their returns.

At the same time, transaction costs help understand trading volume itself. In frictionless models, investors continuously rebalance their portfolios, as to hold a constant mix of assets over time. Since trading volume is proportional to the total variation of a portfolio, and prices follow diffusions that have infinite variation, such models lead to the absurd conclusion that trading volume is infinite over any time interval. With transaction costs, even small trading costs make it optimal for investors to trade infrequently, allowing wide oscillations in their portfolios. 

This paper reviews a recent approach, which has made portfolio choice with transaction costs more tractable, and which appears to be applicable in more complex settings. This approach is not based on any new revolutionary concept, but it rather tries to combine several ideas that were previously used in isolation. Thus, we present a new toolbox that contains several used tools.
In a nutshell, we argue that a natural approach to portfolio choice problems with transaction costs entails four steps: (i) heuristic control arguments to identify the long-run value function, (ii) construction of a candidate shadow price using marginal rates of substitution, (iii) verification and finite-horizon bounds using the myopic probability, and (iv) asymptotic results from the implicit function theorem.

The advantages of this approach are threefold. First, it combines the dimension-reduction and higher tractability of the long-horizon problem with exact finite-horizon bounds, which keep a firm grip on the robustness of the solution. Second, we show that the free-boundaries arising with transaction costs can sometimes be identified explicitly in terms of a single parameter, the \emph{equivalent safe rate}, which remains the only non-explicit part of the solution. This reduction is useful both for theoretical and for practical purposes, as it helps to simplify proofs as well as calculations. Third, this approach leads to the simultaneous computation of several related quantities, such as welfare, portfolios, liquidity premia, and trading volume. 

The paper proceeds as follows: in the next section, we present a brief timeline of related research, which is far from exhaustive, and only aims at putting the paper in context. The following section introduces the main problem, discussing the relative advantages of the three main models with terminal wealth, consumption, and long-horizon. This section also discusses the typical heuristic arguments of stochastic control that lead to an educated guess for the value function, and the identification of the corresponding free boundaries. For the long-run problem, the following section shows the passage from the heuristic calculations to a verification, which relies on two central ideas. The first one is the construction of a \emph{shadow} market, an imaginary frictionless market, built to deliver the same optimal strategy as the original market with transaction cost. This shadow market harnesses transaction costs by hiding them inside a more complex model, without transaction costs, but in which investment opportunities are driven by a state variable that represents the portfolio composition of the investor. This insight -- that transaction costs are essentially equivalent to state-dependent investment opportunities -- in turn allows to exploit the approach to verification based on the change of measure to the myopic probability.

We conclude with a deliberately speculative section on three open problems: multiple assets, return predictability, and option spreads. We argue that with transaction costs, multivariate models present both a substantial technical challenge, and a potentially fertile ground for novel financial insights, which may alter the conventional wisdom on fund separation. Likewise, transaction costs may help reconcile statistical evidence on return predictability with the poor out-of-sample performance of market-timing strategies. Finally, the large bid-ask spreads observed in options on highly liquid assets still lack a theoretical basis, and transaction costs are a natural avenue to search for an explanation.

\section{Literature Review}

Portfolio choice with transaction cost starts with the seminal papers of \citet*{MR0469196}, \citet*{constantinides.86}, and \cite{dumas.luciano.91}, in the wake of the frictionless results of Merton \citeyearpar{merton.69,merton.71}. From heuristic arguments,  these early studies gleaned central insights that held up to subsequent formal proofs. First, optimal portfolios entail a no-trade region, in which it is optimal to keep existing holdings in all assets. Optimal portfolios always remain within this region, and hence trading should merely take place at its boundaries. The no-trade region is wide, even for small transaction costs, implying that investors should accept wide fluctuations around the frictionless target.  Second, the large no-trade region has a small welfare impact \citep*{constantinides.86}, because the displacement loss is small near the frictionless optimum, and the wide no-trade region minimizes the effect of transaction costs.

On the mathematical side, \citet*{MR942619} reduce the maximization of logarithmic utility from terminal wealth at a long horizon to the solution of a nonlinear second-order ODE with free boundaries, to be determined numerically. \citet*{MR1080472} accomplish this feat for power utility from consumption with infinite horizon.  \citet*{MR1284980} extend their analysis with viscosity techniques, removing some parametric restrictions.  \citet*{MR1284980} and \citet*{MR2076549} study the size of the no-trade interval and the utility loss due to transaction costs. They argue that these are of order $O(\eps^{1/3})$ and $O(\eps^{2/3})$, respectively, where $\eps$ is the proportional cost, in line with the numerical results of \cite{constantinides.86} alluded to above.  Building on earlier heuristic results of \cite{whalley.wilmott.97}, \citet*{MR2048827} explicitly determine the coefficients of the leading-order corrections around the frictionless case $\eps=0$. In a general Markovian setting and for arbitrary utility functions, \cite{soner.touzi.12} characterize the corresponding quantities in terms of an ergodic control problem. All these papers employ stochastic control as their main tool.

A new strand of literature, which finds its roots in the seminal work of \citet*{MR1338025} and \citet*{MR1384221}, seeks to bring martingale methods, now well-understood in frictionless markets, to bear on transaction costs. 
This idea has already shown its promise in the context of superreplication: \citet*{MR2398764} prove the face-lifting theorem of \citet*{MR1336872} for general continuous processes, using an argument based on shadow prices.
\citet*{MR2676941} explore this approach for optimal consumption from logarithmic utility, showing how shadow prices simplify verification theorems. \citet*{gerhold2010dual} exploit this idea to obtain the expansions of \citet*{MR2048827} for logarithmic utility, but with an arbitrary number of terms. The present study reviews a streamlined version of the approach put forward by \cite{gerhold.al.11}, who prove a verification theorem and derive full asymptotics for the optimal policy, welfare, and implied trading volume in the long-run model of \cite{dumas.luciano.91}. The duality-based verification is based on applying the frictionless long-run machinery of \cite{guasoni.robertson.11} to a fictitious shadow price, traded without transaction costs. Compared to \cite{gerhold.al.11}, finding a candidate shadow price is greatly simplified by applying a observation originally made by \cite{loewenstein.00}: Given a smooth candidate value function, it can simply be obtained via the marginal rate of substitution of risky for safe assets for the frictional investor. (Also cf.\ \cite{herzegh.prokaj.11,muhlekarbe.liu.12} for applications of this idea to related problems.)

\section{The Basic Model}

\subsection{Objectives}
Let $X^\pi_t$ denote the wealth of an investor who follows the portfolio $\pi_t$, and let $c_t$ his consumption rate, both at time $t$. The three typical objectives for portfolio choice with power utility $U(x)=x^{1-\gamma}/(1-\gamma)$ are:
\begin{align}
\label{eq:termwelath}
\max_\pi\ &\esp{\frac{(X^\pi_T)^{1-\gamma}}{1-\gamma}},
&\qquad\text{(terminal wealth)}
\\
\label{eq:consumption}
\max_{\pi,c}\ &\esp{\int_0^\infty e^{-\delta t} \frac{c_t^{1-\gamma}}{1-\gamma} dt},
&\qquad\text{(consumption)}
\\
\label{eq:longrun}
\max_\pi &\liminf_{T\rightarrow\infty} \frac 1T\log \esp{{(X^\pi_T)^{1-\gamma}}}^{\frac1{1-\gamma}}.
&\qquad\text{(long run)}
\end{align}

Expected utility from terminal wealth \eqref{eq:termwelath} has attracted the attention of most of the semimartingale literature (see, for example, \citet{kramkov1999asymptotic} and the references therein). This objective is the simplest for abstract questions, such as existence, uniqueness, well-posedness, and stability, which are by now largely understood. It is also relevant for problems such as retirement planning, which entail a known horizon and no intermediate consumption.
Expected utility from intertemporal consumption \eqref{eq:consumption} is more appealing for applications to macroeconomics, because it yields an endogenous consumption process $c_t$, and therefore has testable implications for consumption data. 

The long run objective is probably the least intuitive, in view of the limit in  \eqref{eq:longrun}. To understand its economic interpretation, note first that, for a fixed horizon $T$, the quantity $\esp{{(X^\pi_T)^{1-\gamma}}}^{\frac1{1-\gamma}}$ coincides with the \emph{certainty equivalent} $U^{-1}(\esp{U(X^\pi_T)})$ of the payoff $X^\pi_T$. If we match this certainty equivalent with $x e^{\rho_T T}$, that is, the investor's initial capital $x$ compounded at some constant rate $\rho_T$ for the same horizon $T$, we recognize that:
$$
\rho_T = \frac 1T\log \esp{{(X^\pi_T)^{1-\gamma}}}^{\frac1{1-\gamma}}.
$$
Thus, the limit in \eqref{eq:longrun} has the interpretation of an \emph{equivalent safe rate}, that is, the hypothetical safe rate that would make the investor indifferent between investing optimally in the market, and leaving all wealth invested at this hypothetical rate.

Both the consumption and long-run problems are \emph{stationary} objectives, in that they lead to time-independent solutions (as long as investment opportunities are also stationary). Of course, the advantage of stationary problems is that the resulting optimization problems have one less dimension than similar nonstationary problems, such as utility maximization from terminal wealth. Both objectives model an investor with an infinite horizon, but with some important differences. First, the consumption objective involves the additional  time-preference rate $\delta$, which does not appear in the long-run objective. Second, in typical models (even in a Black-Scholes market, compare \cite*{choi.al.12}), the consumption objective may not be well posed if risk aversion $\gamma$ is less than the logarithmic value of one, and investment opportunities are sufficiently attractive. By contrast, the long-run objective is typically well-posed under more general conditions.

The irony of portfolio choice is that its most natural objectives are also the least tractable: the terminal-wealth problem admits closed-form solutions only in rare cases (cf., e.g., \cite{liu.07}). Even when such solutions exist, they are often too clumsy to yield clear insights on the role of preference and market parameters. Unfortunately, the consumption objective admits explicit solutions primarily in complete markets, or with investment opportunities independent of asset prices, a fact that severely limits our understanding of the effects of partial return predictability on consumption.

The good news is that the long-run problem admits explicit solutions in many situations in which the other two problems do not, its optimal portfolio is almost optimal even for the other objectives, and bounds on the resulting utility loss are available. This general insight is crucial in markets with frictions, such as transaction costs.

\subsection{Control Heuristics}

We now examine the differences in the Hamilton-Jacobi-Bellman equations arising from the three objectives \eqref{eq:termwelath}, \eqref{eq:consumption}, and \eqref{eq:longrun}, in the basic model with one safe asset growing at the riskless rate $r \geq 0$, and a risky asset with ask (buying) price $S_t$ following geometric Brownian Motion:
\begin{equation}
\frac{dS_t}{S_t} = (\mu+r) dt + \sigma dW_t, \quad \mu, \sigma>0.
\end{equation}
The bid (selling) price is  $(1-\varepsilon)S_t$, where $\varepsilon \in (0,1)$ is the relative bid-ask spread.

Denote the number of units of the safe asset by $\varphi^0_t$ and write the number of units of the risky asset $\varphi_t=\varphi_t^{\uparrow}-\varphi_t^\downarrow$ as the difference between cumulative purchases and sales. 
The values of the safe position $X_t$ and of the risky position $Y_t$ (quoted at the ask price) evolve as:
\begin{align}
\label{eq:selffin_risky}
dX_t =& r X_t dt-S_t d\varphi^{\uparrow}_t+(1-\ve)S_t d\varphi^{\downarrow}_t,\\
\label{eq:selffin_safe}
dY_t =& (\mu+r) Y_t dt +\sigma Y_t dW_t +S_td\varphi^{\uparrow}_t-S_td\varphi^{\downarrow}_t.
\end{align}
The second equation prescribes that risky wealth earns the return on the risky asset, plus units purchased, and minus units sold. In the first equation the safe position earns the safe rate, minus the units used for purchases (at the ask price $S_t$), and plus the units used for sales (at the bid price $(1-\varepsilon)S_t$).

For the maximization of utility from terminal wealth, denote the value function as $V(t,x,y)$, which depends on time $t$, on the safe position $x$, and on the risky position $y$.
It\^o's formula yields:
\begin{align}
d V(t,X_t,Y_t)=& V_t dt+V_x dX_t + V_y dY_t +\frac 12  V_{yy} d\langle Y,Y\rangle_t\\
=& 
\label{eq:itoder}
\left(V_t+ r X_t V_x+(\mu+r)Y_t V_y+\frac{\sigma^2}2 Y_t^2 V_{yy}\right)dt\\
&+ S_t(V_y-V_x)d\varphi^{\uparrow}_t+S_t((1-\ve)V_x-V_y)d\varphi^{\downarrow}_t+\sigma X_t V_y dW_t,
\end{align}
By the martingale optimality principle of stochastic control, the value function $V(t,X_t,Y_t)$ must be a supermartingale for any choice of purchases and sales $\varphi_t^{\uparrow},\varphi_t^{\downarrow}$. Since these are increasing processes,  this implies $V_y-V_x \le 0$ and $(1-\ve) V_x-V_y \le 0$, which means that
\begin{equation}
1 \le \frac{V_x}{V_y}\le \frac{1}{1-\ve}.
\end{equation}
In the interior of this ``no-trade region'', where the number $\varphi_t=\varphi_t^{\uparrow}-\varphi_t^{\downarrow}$ of risky shares remains constant, the drift of $V(t,X_t,Y_t)$ cannot be positive, and must become zero for the optimal policy. This leads to the HJB equation: 
\begin{equation}\label{eq:hjb}
 V_t+ r X_t V_x+(\mu+r)Y_t V_y+\frac{\sigma^2}2 Y_t^2 V_{yy} =0
 \qquad \text{if } \qquad 1< \frac{V_x}{V_y}<\frac{1}{1-\ve}.
\end{equation}
Next, the value function is homogeneous in wealth, i.e. $V(t,X_t,Y_t)=(X_t)^{1-\gamma}v(t,Y_t/X_t)$, whence setting $z=y/x$: 
\begin{equation}\label{eq:hjbred}
\frac{\sigma^2}2 z^2 v_{zz}+\mu z v_z+r(1-\gamma)v+v_t=0
\qquad \text{if } \qquad 1+z<\frac{ (1-\gamma) v(t,z)}{v_z(t,z)}<\frac{1}{1-\ve}+z.
\end{equation}
Now, suppose that the no-trade region $\{(t,z):1+z\leq \frac{ (1-\gamma) v(t,z)}{v_z(t,z)}\leq\frac{1}{1-\ve}+z\}$ coincides with some interval $l(t)\le z\le u(t)$ to be found. At $l(t)$ the left inequality in \eqref{eq:hjbred} holds as equality, while at $u(t)$ the right inequality holds as equality, leading to the boundary conditions:
\begin{align}
\label{eq:boundbuy}
(1+l)v_z(t,l)-(1-\gamma)v(t,l)&=0,\\
\label{eq:boundsell}
(1/(1-\ve)+u)v_z(t,u)-(1-\gamma)v(t,u)&=0.
\end{align}
These conditions are not sufficient to identify the solution to the optimization problem, since they can be matched for any trading boundary $l(t), u(t)$. The optimal boundaries are identified as the ones that satisfy the smooth-pasting conditions. These conditions can be seen as limits of the optimality conditions for an impulse control problem with a infinitesimally small cost (\cite{dumas.91}). In practice, they are derived by differentiating \eqref{eq:boundbuy} and \eqref{eq:boundsell} with respect to $z$ at the respective boundaries $z=l$ and $z=u$:
\begin{align}
\label{eq:smoothbuy}
(1+l)v_{zz}(t,l)+\gamma v_z(t,l)=0,\\
(1/(1-\ve)+u)v_{zz}(t,u)+ \gamma v_z(t,u)=0.
\label{smoothsell}
\end{align}
This system defines a \emph{two-dimensional, linear} free-boundary problem in $(t,z)\in [0,T]\times \mathbb R$, which is not tractable in general. \citet{liu2002optimal} obtain a semiexplicit solution with the randomization approach used by \citet{carr1998randomization} to price American options.

With utility maximization from infinite-horizon consumption, the value function depends only on the safe and risky positions $X_t$, $Y_t$ -- the problem is stationary. Calculations are similar, with some minor differences: first, the self-financing condition \eqref{eq:selffin_safe} must include the term $-c_t dt$ in the cash balance, to account for consumption expenditures. Then, the martingale optimality principle takes the following slightly different form. Since utility is not only incurred at maturity but from consumption along the way, not the value function itself but the sum of past consumption $\int_0^t e^{-\delta u} c_u^{1-\gamma}/(1-\gamma) du$ and the value function $e^{-\delta t} V(X_t,Y_t)$, representing future consumption, should be a supermartingale for any policy and a martingale for the optimizer. Here, $\delta$ is the time-preference parameter in \eqref{eq:consumption}. Then, the term $V_t$ in \eqref{eq:itoder} and in turn \eqref{eq:hjb} has to be replaced by $\frac{c_t^{1-\gamma}}{1-\gamma}-c_t V_x- \delta V$. Pointwise maximization yields the optimal consumption rate $c_t=V_x^{-1/\gamma}$. Plugging this expression back into the HJB equation and accounting for homogeneity in wealth, the corresponding free-boundary problem for the reduced value function $v(z)$ then reads as: 
\begin{equation}\label{eq:hjbcons}
\frac{\sigma^2}2 z^2 v_{zz}+
\mu z v_z  +
((1-\gamma)r- \delta) v+
\frac{\gamma}{1-\gamma} ((1-\gamma)v-z v_z)^{1-\frac1\gamma}
=0
\quad \text{if } \quad 1+z<\frac{ (1-\gamma) v(z)}{v_z(z)}<\frac{1}{1-\ve}+z.
\end{equation}
This is the \emph{one-dimensional, nonlinear} free-boundary problem studied by \cite{MR1080472}, who prove a verification theorem, and find a numerical solution. Still, this problem is nontrivial, because the free boundaries points $l, u$ are not easy to identify in terms of the model parameters, and the second order, nonlinear equation \eqref{eq:hjbcons} does not admit a known explicit solution for given initial conditions.

The long-run problem \eqref{eq:longrun} gives the best of both worlds, and more. 
But it requires more audacious heuristics, and nonstandard arguments to be made precise. Puzzlingly enough, this approach was proposed very early in the transaction costs literature by \citet*{MR942619} and \citet*{dumas.luciano.91}, but its potential has not become clear until recently.

We start from equation \eqref{eq:hjb}, derived for a fixed horizon $T$, and note that the value function $V$ should grow exponentially with the horizon. This observation, combined with homogeneity in wealth, leads to guess a solution of the form 
$V(t,X_t,Y_t)= (X_t)^{1-\gamma}v(Y_t/X_t) e^{- (1-\gamma)  (r+\beta) t}$. It is clear that such a guess in general does not solve the finite-horizon problem, as it fails to satisfy its terminal condition. But it is reasonable to expect that it governs the long-run problem, for which the horizon never approaches. With the above guess, the HJB equation reduces to
\begin{equation}\label{eq:hjblong}
\frac{\sigma^2}2 z^2v''(z)+\mu z v'(z)- (1-\gamma) \beta v(z)=0
\qquad \text{if } \qquad 1+z<\frac{ (1-\gamma) v(z)}{v'(z)}<\frac{1}{1-\ve}+z,
\end{equation}
which is a \emph{one-dimensional, linear} free-boundary problem -- the best of both worlds. Note that in this system $\beta$ is not exogenous, but an unknown parameter that determines the growth rate of the value function, and which has to be found along with the free boundaries $l, u$. Indeed, $r+\beta$ is the \emph{equivalent safe rate} which makes a long-term investor indifferent between the original market and this alternative rate alone.

The two crucial advantages of the long-run problem are that the free boundaries $l, u$ have explicit formulas in terms of $\beta$, and that it reduces to solving a Cauchy problem for a \emph{first-order} ordinary differential equation. Depending on the problem at hand, this equation may even have an explicit solution, a fact that is useful although not essential for asymptotics. To derive the free boundary $l$, substitute first 
\eqref{eq:smoothbuy} and then \eqref{eq:boundbuy} into \eqref{eq:hjblong} to obtain 
\begin{align*}
-\frac{\sigma^2}2 (1-\gamma)\gamma\frac{l^2}{(1+l)^2}v +\mu (1-\gamma)\frac{l}{1+l}v -(1-\gamma)\beta v=0.
\end{align*}
Now, observe that $\pi_-=l/(1+l)$ is precisely the risky portfolio weight at the buy boundary, evaluated at the ask price. Factoring out $(1-\gamma) v$, it follows that (cf.  \citet*{dumas.luciano.91})
\begin{align}\label{eq:dumluceq}
-\frac{\gamma \sigma^2}2 \pi_-^2+\mu \pi_- -\beta=0.
\end{align}
Likewise, a similar calculation for $u$ shows that the other root of \eqref{eq:dumluceq} is $\pi_+=u (1-\ve)/(1+u (1-\ve))$, which coincides with the risky portfolio weight at the sell boundary, evaulated at the bid price.

After these calculations, the boundaries $l, u$, or equivalently $\pi_-, \pi_+$, are uniquely identified as solutions of the above equation, once the parameter $\beta$ is found:
$$\pi_{\pm}=\frac{\mu \pm \sqrt{\mu^2-2\gamma\sigma^2\beta}}{\gamma\sigma^2}.$$
 The formulas become even clearer by replacing the parameter $\beta$ with $\lambda=\sqrt{\mu^2-2\gamma\sigma^2\beta}$. With this notation, in which $\lambda=0$ corresponds to the frictionless setting, $\beta=(\mu^2-\lambda^2)/2\gamma\sigma^2$ and the buy and sell boundaries have the intuitive representation 
\begin{align}\label{eq:uplow}
\pi_\pm=\frac{\mu\pm\lambda}{\gamma \sigma^2},
\end{align}
from which $l=\frac{\pi_{-}}{1-\pi_{-}}$ and $u=\frac{1}{1-\varepsilon}\frac{\pi_+}{1-\pi_+}$ are obtained directly. Thus, it remains to find $\lambda$ to identify both the free-boundaries and the equivalent safe rate $r+\beta$. To this end, it is convenient to apply the substitution
\[
  v(z)=e^{(1-\gamma)\int_0^{\log(z/l(\lambda))} w(y)dy}, \quad \mbox{i.e.,} \quad w(y)=\frac{l(\lambda)e^y v'(l(\lambda)e^y)}{(1-\gamma)v(l(\lambda)e^y)}, \]
which reduces the free-boundary problem to a Cauchy problem with a terminal condition:
\begin{align}
w'(y)+(1-\gamma)w(y)^2+\left(\frac{2\mu}{\sigma^2}-1\right)w(y)-
\gamma \left(\frac{\mu-\lambda}{\gamma \sigma^2}\right)
\left(\frac{\mu+\lambda}{\gamma \sigma^2}\right)
 &= 0, \quad y\in[0,\log u(\lambda)/l(\lambda)], \label{riccati1}\\
w(0) &= \frac{\mu-\lambda}{\gamma \sigma^2},\label{riccati2}\\
w(\log (u(\lambda)/l(\lambda))) &= \frac{\mu+\lambda}{\gamma \sigma^2}.\label{riccati3}
\end{align}
In other words, the correct value of $\lambda$ is identifed as the one for which the above first-order Riccati equation satisfies both the initial and the terminal value conditions. For a fixed $\ve$, such a value is the solution of a scalar equation obtained from the explicit solution of the Riccati equation (cf.\ Lemma~\ref{lem:riccati}). For $\ve\sim 0$, the asymptotic expansion of $\lambda(\ve)$ follows from the implicit function theorem -- and some patient calculations (see Lemma \ref{lem:lambda} below). 

Now, one could argue that the advantage of the nonlinear, first-order equation \eqref{riccati1} over the linear, second-order equation \eqref{eq:hjblong} is only marginal. In fact, the variable $w$ has the additional advantage, albeit still hidden at this point, that it coincides with the optimal \emph{shadow} risky portfolio weight (cf.\ Lemma \ref{lem:ergodic}), a fact that is hinted at by its boundary conditions. Furthermore, and as a result, for $\gamma=1$ the equation \eqref{riccati1} recovers the case of logarithmic utility \citep{gerhold2010dual}, while \eqref{eq:hjblong} does not.

\subsection{Explicit formulas}

Let us now show that the reduced value function $w$ and the quantity $\lambda$ are indeed well-defined. To this end, first determine, for a given small $\lambda>0$, an explicit expression for the solution $w$ of the ODE~\eqref{riccati1}, complemented by the initial condition~\eqref{riccati2}.

\begin{lemma}\label{lem:riccati}
Let $0<\mu/\gamma \sigma^2 \neq 1$.
Then for sufficiently small $\lambda>0$, the function
$$
w(\lambda,y)=\begin{cases} 
\frac{a(\lambda)\tanh[\tanh^{-1}(b(\lambda)/a(\lambda))-a(\lambda)y]+(\frac{\mu}{\sigma^2}-\frac{1}{2})}{\gamma-1}, &\mbox{if } \gamma \in (0,1) \mbox{ and } \frac{\mu}{\gamma\sigma^2}<1 \mbox{ or } \gamma>1 \mbox{ and } \frac{\mu}{\gamma\sigma^2}>1,\\
\frac{a(\lambda) \tan[\tan^{-1}(b(\lambda)/a(\lambda))+a(\lambda)y]+(\frac{\mu}{\sigma^2}-\frac{1}{2})}{\gamma-1}, &\mbox{if } \gamma>1 \mbox{ and } \frac{\mu}{\gamma \sigma^2} \in \left(\frac{1}{2}-\frac{1}{2}\sqrt{1-\frac{1}{\gamma}},\frac{1}{2}+\frac{1}{2}\sqrt{1-\frac{1}{\gamma}}\right),\\
\frac{a(\lambda)\coth[\coth^{-1}(b(\lambda)/a(\lambda))-a(\lambda)y]+(\frac{\mu}{\sigma^2}-\frac{1}{2})}{\gamma-1}, &\mbox{otherwise},
\end{cases}
$$  
with
\[
a(\lambda)=\sqrt{\Big|(\gamma-1)\frac{\mu^2-\lambda^2}{\gamma\sigma^4}-\Big(\frac{1}{2}-\frac{\mu}{\sigma^2}\Big)^2\Big|} \quad \mbox{and} \quad b(\lambda)=\frac{1}{2}-\frac{\mu}{\sigma^2}+(\gamma-1)\frac{\mu-\lambda}{\gamma\sigma^2},
\]
is a local solution of 
\begin{equation}\label{eq:wode}
w'(y)+(1-\gamma)w^2(y)+\left(\frac{2\mu}{\sigma^2}-1\right)w(y)-\frac{\mu^2-\lambda^2}{\gamma \sigma^4}=0, \quad w(0)=\frac{\mu-\lambda}{\gamma\sigma^2}.
\end{equation}
Moreover, $y \mapsto w(\lambda,y)$ is increasing (resp.\ decreasing) for $\mu/\gamma\sigma^2 \in (0,1)$ (resp.\ $\mu/\gamma\sigma^2>1$).
\end{lemma}

\begin{proof}
The first part of the assertion is easily verified by taking derivatives. The second follows by inspection of the explicit formulas.
\end{proof}

Next, establish that the crucial constant $\lambda$, which determines both the no-trade region and the equivalent safe rate, is well-defined. For small transaction costs $\varepsilon \sim 0$, its asymptotics are readily computed by means of the implicit function theorem.

\begin{lemma}\label{lem:lambda}
Let $0<\mu/\gamma \sigma^2 \neq 1$ and $w(\lambda,\cdot)$ be defined as in
Lemma~\ref{lem:riccati}, and set
$$
l(\lambda)=\frac{\mu-\lambda}{\gamma\sigma^2-(\mu-\lambda)}, \quad u(\lambda)=\frac{1}{(1-\ve)}\frac{\mu+\lambda}{\gamma\sigma^2-(\mu+\lambda)}.
$$
Then, for sufficiently small $\ve>0$, there exists a unique solution $\lambda$ of 
\begin{equation}\label{eq:wrbd}
  w\left(\lambda,\log\left(\frac{u(\lambda)}{l(\lambda)}\right)\right)
  -\frac{\mu+\lambda}{\gamma\sigma^2}=0.
\end{equation}
As $\ve \downarrow 0$, it has the asymptotics
\[
  \lambda = \gamma\sigma^2\left(\frac{3}{4\gamma}\left(\frac{\mu}{\gamma\sigma^2}\right)^2\left(1-\frac{\mu}{\gamma\sigma^2}\right)^2\right)^{1/3}\ve^{1/3}+O(\ve).
\]
\end{lemma}
\begin{proof}
Write the boundary condition~\eqref{eq:wrbd} as $f(\lambda,\eps)=0$, where:
\begin{equation*}\label{eq:flameps}
f(\lambda,\eps) = w(\lambda, \log(u(\lambda)/l(\lambda)))-\frac{\mu+\lambda}{\gamma \sigma^2}.
\end{equation*}
Of course, $f(0,0)=0$ corresponds to the frictionless case. The implicit function theorem then suggests that for sufficiently small $\varepsilon$ there exists a unique zero $\lambda(\eps)$ with the asymptotics $\lambda(\eps) \sim - \eps f_\eps/f_\lambda$, but the difficulty is that $f_\lambda=0$, because $\lambda$ is not of order $\eps$. Heuristic arguments \citep*{MR1284980,MR2076549} suggest that $\lambda$ is of order $\eps^{1/3}$. Thus, setting $\lambda=\delta^{1/3}$ and $\hat f(\delta,\eps)=f(\delta^{1/3},\eps)$, and computing the derivatives of the explicit formula for $w(\lambda,x)$ (cf.\ Lemma \ref{lem:riccati}) shows that: 
\begin{align*}
\hat f_\eps(0,0) = -\frac{\mu  \left(\mu -\gamma  \sigma ^2\right)}{\gamma ^2 \sigma ^4}, \qquad \hat f_\delta(0,0) = \frac{4}{3 \mu ^2 \sigma ^2-3 \gamma  \mu  \sigma ^4}.
\end{align*}
As a result:
\begin{equation*}
\delta(\eps) \sim -\frac{\hat{f}_\eps(0,0)}{\hat{f}_\delta(0,0)} \eps = 
\frac{3 \mu ^2 \left(\mu -\gamma  \sigma^2\right)^2}{4 \gamma ^2 \sigma ^2}\eps
\quad\text{whence}\quad
\lambda(\eps) \sim
\left(\frac{3\mu ^2 \left(\mu -\gamma  \sigma ^2\right)^2}{4\gamma ^2 \sigma ^2}\right)^{1/3} \eps^{1/3}.
\end{equation*}
\end{proof}

Henceforth, consider small transaction costs $\ve>0$,
and let~$\lambda$ denote the constant in Lemma~\ref{lem:lambda}. Moreover, set $w(y)=w(\lambda, y)$, $a=a(\lambda)$, $b=b(\lambda)$, and $u=u(\lambda)$, $l=l(\lambda)$. In all cases, the function $w$ can be extended smoothly to an open neighborhood of $[0,\log(u/l)]$ (resp.\ $[\log(u/l),0]$ if $\mu/\gamma\sigma^2>1$). By continuity, the ODE \eqref{eq:wode} then also holds at $0$ and $\log(u/l)$; inserting the boundary conditions for $w$ yields the following counterparts for the derivative $w'$:

\begin{lemma}\label{lem:smoothpasting}
Let  $0<\mu/\gamma \sigma^2 \neq 1$. Then, in all three cases,
$$w'(0)=\frac{\mu-\lambda}{\gamma\sigma^2}-\left(\frac{\mu-\lambda}{\gamma\sigma^2}\right)^2, \quad w'\left(\log\left(\frac{u}{l}\right)\right)=\frac{\mu+\lambda}{\gamma\sigma^2}-\left(\frac{\mu+\lambda}{\gamma\sigma^2}\right)^2.$$
\end{lemma}

\subsection{Discussion}

The above heuristics offer a practical approach to portfolio choice problems with transaction costs, and can be adapted to accommodate additional model features. More importantly, they yield results that are robust to the model specification. In view of \eqref{eq:uplow} and the asymptotics for $\lambda$ in Lemma \ref{lem:lambda}, the no-trade boundaries have the expansion:
\begin{align}\label{eq:notrade}
\pi_{\pm} &= \frac{\mu}{\gamma \sigma^2} \pm 
\left(\frac{3}{4\gamma} 
\left(\frac{\mu}{\gamma \sigma^2}\right)^2\left(1-\frac{\mu}{\gamma \sigma^2}\right)^2\right)^{1/3} \ve^{1/3} +O(\ve).
\end{align}
This expansion coincides with the one obtained by \cite{MR2048827} in the model with consumption. In other words, the long-run and the consumption models yield exactly the same solution at the leading order for small transaction costs. The expansions do differ at the second order, but such differences tend to have a modest effect for typical parameter values.

A major advantage of the long-run objective is the possibility to reduce the solution to a single algebraic equation for the parameter $\lambda$, in terms of which the free boundaries are found explicitly. In principle, one could attempt the same reduction in the consumption problem, by substituting equations \eqref{eq:boundbuy} and \eqref{eq:smoothbuy} into \eqref{eq:hjbred}. The result is a scalar equation for $l$ in terms of $v(l)$, the value of the reduced value function at the trading boundary. Alas, the equation does not have an explicit solution. Also, the value of $v(l)$ is identified as the only one for which the solution to the differential equation (which also has no explicit solution) matches the analogous boundary condition at $u$. The situation is disappointingly more complicated than \eqref{eq:dumluceq}, which immediately identifies both boundaries in terms of a single parameter. In summary, the consumption problem yields a solution which is strikingly similar to the long-run problem, but in a much less tractable setting. Vice versa, the long-run solution provides a  tractable first-order approximation to the consumption problem. 

In the same vein, the long-run optimal portfolio is not far from optimal for utility maximization with terminal wealth. Indeed, \cite{gerhold.al.11} show that the wealth corresponding to the long-run optimal portfolio matches the value function for any \emph{finite} horizon $T$ \citep{bichuch.11} at the leading order $\eps^{2/3}$ for small transaction costs $\varepsilon$ (compare Theorem \ref{th:finhor} below). Hence, finite horizons -- like consumption -- only have a second-order effect on portfolio choice with transaction costs.

To apply the heuristic steps above to more complex problems with transaction costs, it is worth distinguishing the aspects that are special to the specific problem at hand from the ones that are flexible enough to be useful in other models. 
First, in general one cannot expect that a single, simple equation like \eqref{eq:dumluceq} identifies both free boundaries. But the same argument that leads to this equation (the substitution of the boundary and smooth pasting conditions into the HJB equation) will generally lead in a long-run problem to some scalar equation for each boundary, in terms of the equivalent safe rate $\beta$ of the problem. Such equations may be solved explicitly (as in the case of \eqref{eq:dumluceq}) or not, but in the latter case an asymptotic solution will still be available, expanding the scalar equation around the frictionless values of $(\pi,\beta)$. 

Second, the reduced HJB equation may not be autonomous or have an explicit solution, which are two special features of \eqref{riccati1}. If the equation is not autonomous, the Cauchy problem cannot be started at some arbitrary point (zero in the previous example) without a further change of variable. If the free boundaries admit explicit solutions in terms of $\beta$, sometimes a careful choice of notation can lead to a simple expression for at least one boundary, which is a natural choice for the starting point of the Cauchy problem. The correct value of $\beta$ is then identified as the one for which the remaining boundary condition is satisfied. Even if the differential equation has an explicit solution, this condition in general involves a scalar equation that cannot be solved explicitly. Regardless of an explicit formula, asymptotic expansions can be derived by substituting a series expansion for $w$ in the differential equation.

\section{Shadow Prices and Verification}

We justify the heuristic arguments in the previous section by reducing the portfolio choice problem with transaction costs to another portfolio choice problem, without transaction costs. To do so, the bid and ask prices are replaced by a single ``shadow price'' $\tilde{S}_t$ evolving within the bid-ask spread, which yields the same optimal policy and utility. Evidently, \emph{any} frictionless market extension with values in the bid-ask spread leads to more favorable terms of trade than the original market with transaction costs. To achieve equality, the particularly unfavorable shadow price must match the trading prices whenever its optimal policy transacts. The latter is then also feasible and in turn optimal in the original market with transaction costs, motivating the following notion.

\begin{definition}\label{defi:shadow}
A \emph{shadow price} is a frictionless price process $\tilde{S}_t$ evolving within the bid-ask spread $((1-\varepsilon)S_t \leq \tilde{S}_t \leq S_t$ a.s. for all $t$), such that there is an optimal strategy for $\tilde{S}_t$ which is of finite variation and entails buying only when the shadow price $\tilde{S}_t$ equals the ask price $S_t$, and selling only when $\tilde{S}_t$ equals the bid price $(1-\varepsilon)S_t$.
\end{definition}

Once a candidate for such a shadow price is identified, long-run verification results for frictionless models (cf. Guasoni and Robertson (2012)) deliver the optimality of the guessed policy. 

\subsection{Derivation of a Candidate Shadow Price}

With a smooth candidate value function at hand, a candidate shadow price is identified as follows. By definition, trading the shadow price should not allow the investor to outperform the original market with transaction costs. In particular, if $\tilde{S}_t$ is the value of the shadow price at time $t$, then allowing the frictional investor to carry out at single trade at time $t$ at this \emph{frictionless} price should not allow her to increase her utility. A trade of $\nu$ risky shares at the frictionless price $\tilde{S}_t$ moves the investor's safe position $X_t$ to $X_t-\nu \tilde{S}_t$ and her risky position (valued at the ask price $S_t$) from $Y_t$ to $Y_t+\nu S_t$. Then -- recalling that the second and third arguments of the candidate value functions $V$ from the previous section were precisely the investor's safe and risky positions -- the requirement that such a trade does not increase the investor's utility is tantamount to:
$$V(t,X_t-\nu \tilde{S}_t,Y_t+\nu S_t) \leq V(t,X_t,Y_t), \quad \forall \nu \in \mathbb{R}.$$
A Taylor expansion of the left-hand side for small $\nu$ then implies that $-\nu \tilde{S}_t V_x+\nu S_t V_y \leq 0$. Since this inequality has to hold both for positive and negative values of $\nu$, it implies that
\begin{equation}\label{eq:mrs}
\tilde{S}_t=\frac{V_y}{V_x} S_t.
\end{equation}
That is, the multiplicative deviation of the  shadow price from the ask price should be the marginal rate of substitution of risky for safe assets for the optimal frictional investor. In particular, this formula immediately yields a candidate shadow price, once a smooth candidate value function has been identified. For the long-run problem, we derived the following candidate value function in the previous section: 
$$V(t,X_t,Y_t)=e^{-(1-\gamma)(r+\beta)t}(X_t)^{1-\gamma} e^{(1-\gamma)\int_0^{\log(Y_t/lX_t)}w(y)dy}.$$
Using this equality to calculate the partial derivatives in \eqref{eq:mrs}, the candidate shadow price becomes:
\begin{equation}\label{eq:defshadow}
\tilde{S}_t=\frac{w(\Upsilon_t)}{le^{\Upsilon_t}(1-w(\Upsilon_t))}S_t,
\end{equation}
where $\Upsilon_t=\log(X_t/lX^0_t)$ denotes the logarithm of the stock-cash ratio, centered in its value at the lower buying boundary $l$. If this candidate is indeed the right one, then its optimal strategy and value function should coincide with their frictional counterparts derived heuristically above. In particular, the optimal risky fraction $\tilde{\pi}_t$ should correspond to the same numbers $\varphi^0_t$ and $\varphi_t$ of safe and risky shares, but now measured in terms of $\tilde{S}_t$ instead of the ask price $S_t$. As a consequence:
\begin{equation}\label{eq:tildew}
\tilde{\pi}_t=\frac{\varphi_t \tilde{S}_t}{\varphi^0_tS^0_t+\varphi_t\tilde{S}_t}=\frac{\varphi_t S_t \frac{w(\Upsilon_t)}{le^{\Upsilon_t}(1-w(\Upsilon_t))}}{\varphi^0_t S^0_t +\varphi_t S_t \frac{w(\Upsilon_t)}{le^{\Upsilon_t}(1-w(\Upsilon_t))}}=\frac{\frac{w(\Upsilon_t)}{1-w(\Upsilon_t)}}{1+\frac{w(\Upsilon_t)}{1-w(\Upsilon_t)}}=w(\Upsilon_t),
\end{equation}
where, for the third equality, we have used that the optimal frictional stock-cash ratio $\varphi_t S_t/\varphi^0_t S^0_t$ equals $le^{\Upsilon_t}$ by definition of $\Upsilon_t$. We now turn to the corresponding value function $\tilde{V}$. By the definition of shadow price, it should coincide with its frictional counterpart $V$. In the frictionless case, it is more convenient to factor out the total wealth $\tilde{X}_t=\varphi^0_t S^0_t+\varphi_t \tilde{S}_t$ (in terms of the frictionless risky price $\tilde{S}_t$) instead of the safe position $X_t=\varphi^0_t S^0_t$, giving
$$\tilde{V}(t,\tilde{X}_t,\Upsilon_t)=V(t,X_t,Y_t)=e^{-(1-\gamma)(r+\beta)t} \tilde{X}_t^{1-\gamma} \left(\frac{X_t}{\tilde{X}_t}\right)^{1-\gamma} e^{(1-\gamma)\int_0^{\Upsilon_t}w(y)dy}.$$
Since $X_t/\tilde{X}_t=1-w(\Upsilon_t)$ by definition of $\tilde{S}_t$, one can rewrite the last two factors as
\begin{align*}
\left(\frac{X_t}{\tilde{X}_t}\right)^{1-\gamma} e^{(1-\gamma)\int_0^{\Upsilon_t}w(y)dy}&=\exp\left((1-\gamma)\left[\log(1-w(\Upsilon_t))+\int_0^{\Upsilon_t} w(y)dy\right]\right)\\
&=(1-w(0))^{\gamma-1} \exp\left((1-\gamma)\int_0^{\Upsilon_t} \left(w(y)-\frac{w'(y)}{1-w(y)}\right) dy\right).
\end{align*}
Then, setting $\tilde{w}=w-\frac{w'}{1-w}$, the candidate long-run value function for $\tilde{S}$ becomes
$$\tilde{V}(t,\tilde{X}_t,\Upsilon_t)=e^{-(1-\gamma)(r+\beta)t} \tilde{X}_t^{1-\gamma} e^{(1-\gamma)\int_0^{\Upsilon_t}\tilde{w}(y)dy}(1-w(0))^{\gamma-1}.$$
Starting from the candidate value function and optimal policy for $\tilde{S}$, we can now proceed to verify that they are indeed optimal for $\tilde{S}_t$, by adapting the argument from \cite{guasoni.robertson.11}. But before we do that, we have to construct the respective shadow processes.

\subsection{Construction of the Shadow Price}

The above heuristic arguments suggest that the optimal stock-cash ratio $Y_t/X_t=\varphi_t S_t/\varphi^0_t S^0_t$ should take values in the interval $[l,u]$. Hence, $\Upsilon_t=\log(Y_t/lX_t)$ should be $[0,\log(u/l)]$-valued if the lower trading boundary $l$ for the stock-cash ratio $X_t/X^0_t$ is positive. If the investor shorts the safe asset to leverage her risky position, the stock-cash ratio becomes negative. In the frictionless case, and also for small transaction costs, this happens if the Merton proportion $\mu/\gamma\sigma^2$ is bigger than $1$. Then, the trading boundaries $l \leq u$ are both negative, so that the centered log-stock-cash ratio $\Upsilon_t$ should take values in $[\log(u/l),0]$. In both cases, trading should only take place when the stock-cash ratio reaches the boundaries of this region. Hence, the numbers of safe and risky units $\varphi^0_t$ and $\varphi_t$ should remain constant and $\Upsilon_t=\log(\varphi_t/l\varphi^0_t)+\log(S_t/S^0_t)$ should follow a Brownian motion with drift as long as $\Upsilon_t$ moves in $(0,\log(u/l))$ (resp.\ in $(\log(u/l),0)$ if $\mu/\gamma\sigma^2>1$). This motivates to \emph{define} the process $\Upsilon_t$ as reflected Brownian motion:
\begin{equation}\label{eq:defy}
d\Upsilon_t=(\mu-\sigma^2/2)dt+\sigma dW_t+dL_t-dU_t, \quad \Upsilon_0 \in [0,\log(u/l)],
\end{equation}
for continuous, adapted local time processes $L$ and $U$ which are nondecreasing (resp.\ nonincreasing if $\mu/\gamma\sigma^2>1$) and increase (resp.\ decrease if $\mu/\gamma\sigma^2>1$) only on the sets $\{\Upsilon_t=0\}$ and $\{\Upsilon_t=\log(u/l)\}$, respectively. Starting from this process, whose existence is a classical result of \cite{skorokhod.62}, the process $\tilde{S}$ is defined in accordance with \eqref{eq:defshadow}:

\begin{lemma}\label{lem:dynamics}
Let $(\xi^0,\xi) \in \mathbb{R}_+^2$ be the investor's initial endowment in units of the safe and risky asset. Define 
\begin{equation}\label{y0}
y = \begin{cases}
0,& \text{if } l\xi^0 S_0^0 \geq \xi S_0,\\
\log{(u/l)},& \text{if } u\xi^0 S_0^0 \leq \xi S_0,\\
\log{\left[\xi S_0/(\xi^0 S_0^0 l)\right]},& \text{otherwise},
\end{cases}
\end{equation}
and let $\Upsilon$ be defined as in~\eqref{eq:defy}, starting at $\Upsilon_0 = y$. Then, $\tilde S = S \frac{w(\Upsilon)}{l e^{\Upsilon} (1-w(\Upsilon))}$, with $w$ as in Lemma~\ref{lem:riccati}, has the dynamics
\begin{displaymath}
d\tilde S (\Upsilon_t)/\tilde S (\Upsilon_t) = \left(\tilde{\mu}(\Upsilon_t)+r\right)d t+ \tilde{\sigma}(\Upsilon_t)d W_t,
\end{displaymath}
where $\tilde \mu(\cdot)$ and $\tilde \sigma (\cdot)$ are defined as
\begin{displaymath}
\tilde{\mu}(y) = \frac{\sigma^2w'(y)}{w(y)(1-w(y))}\left(\frac{w'(y)}{1-w(y)}-(1-\gamma)w(y)\right), \quad \tilde{\sigma}(y) = \frac{\sigma w'\left(y\right)}{w(y)(1-w(y))}.
\end{displaymath}
Moreover, the process $\tilde S$ takes values within the bid-ask spread $[(1-\varepsilon)S,S]$.
\end{lemma}
Note that the first two cases in~\eqref{y0} arise if the initial stock-cash ratio $\xi S_0/(\xi^0 S_0^0)$ lies outside of the interval $[l,u]$. Then, a jump from the initial position $(\varphi_{0^-}^0, \varphi_{0^-}) = (\xi^0,\xi)$ to the nearest boundary value of $[l,u]$ is required. This transfer requires the purchase resp. sale of the risky asset and hence the initial price $\tilde S _0$ is defined to match the buying resp.\ selling price of the risky asset.

\begin{proof}
The dynamics of $\tilde S_t$ result from It\^{o}'s formula, the dynamics of $\Upsilon_t$, and the identity
\begin{equation}\label{w2ableitung}
w''(y) = 2(\gamma-1)w'(y) w(y)- (2\mu/\sigma^2-1) w'(y),
\end{equation}
obtained by differentiating the ODE \eqref{eq:wode} for $w$ with respect to $x$. Therefore it remains to show that $\tilde{S}_t$ indeed takes values in the bid-ask spread $[(1-\ve)S_t,S_t]$. To this end, notice that -- in view of the ODE \eqref{eq:wode} for $w$ -- the derivative of the function $g(y):=w(y)/l e^y (1-w(y))$ is given by
$$g'(y)=\frac{w'(y)-w(y)+w^2(y)}{le^y (1-w(y))^2}=\frac{\gamma(w^2-2\frac{\mu}{\gamma\sigma^2} w)+(\mu^2-\lambda^2)/\gamma\sigma^4}{le^y (1-w(y))^2}.$$
Due to the boundary conditions for $w$, the derivative $g'$ vanishes at $0$ and $\log(u/l)$. Differentiating its numerator gives $2\gamma w'(y)(w(y)-\frac{\mu}{\gamma\sigma^2})$. For $\frac{\mu}{\gamma\sigma^2} \in (0,1)$ (resp.\ $\frac{\mu}{\gamma\sigma^2}>1$), $w$ is increasing from $\frac{\mu-\lambda}{\gamma\sigma^2}<\frac{\mu}{\gamma\sigma^2}$ to $\frac{\mu+\lambda}{\gamma\sigma^2}>\frac{\mu}{\gamma\sigma^2}$ on $[0,\log(u/l)]$ (resp.\ decreasing from $\frac{\mu+\lambda}{\gamma\sigma^2}$ to $\frac{\mu-\lambda}{\gamma\sigma^2}$ on $[\log(u/l),0]$); hence, $w'$ is nonnegative (resp.\ nonpositive). Moreover, $g'$ starts at zero for $y=0$ (resp.\ $\log(u/l)$), then decreases (resp.\ increases), and eventually starts increasing (resp.\ decreasing) again, until it reaches level zero again for $y=\log(u/l)$ (resp.\ $y=0$). In particular, $g'$ is nonpositive (resp.\ nonnegative), so that $g$ is decreasing on $[0,\log(u/l)]$ (resp.\ increasing on $[\log(u/l),0]$ for $\frac{\mu}{\gamma\sigma^2}>1$). Taking into account that $g(0)= 1$ and $g(\log(u/l))=1-\varepsilon$, by the boundary conditions for $w$ and the definition of $u$ and $l$ in Lemma \ref{lem:lambda}, the proof is now complete.
\end{proof}

\subsection{Verification}

The long-run optimality of the candidate risky weight $\tilde{\pi}(\Upsilon_t)=w(\Upsilon_t)$ from \eqref{eq:tildew} in the frictionless market with price process $\tilde{S}_t$ can now be verified by adapting the argument in \citet*{guasoni.robertson.11}. 
The first step is to determine finite-horizon bounds, which provide lower and upper estimates for the maximal expected utility on any horizon $T$, by focusing on the values of the candidate long-run optimal policy and long-run optimal martingale measure.

These bounds are based on the concept of the (long-run) myopic probability, the hypothetical probability measure under which a logarithmic investor would adopt the same policy as the original power investor under the physical probability. The advantage of this probability is to decompose expected \emph{power} utility (and its dual) into a long-run component \emph{times} a transient component.
This decomposition is similar in spirit to the separation of \emph{logarithmic} utility into a long-run component \emph{plus} a transitory component. To see the analogy, consider the logarithmic utility of a portfolio $\pi(\Theta_t)$ traded in a frictionless market with expected excess return $\tilde\mu(\Theta_t)$ and volatility $\tilde\sigma(\Theta_t)$ driven by some state variable $\Theta_t$:
\begin{equation*}
\log \tilde{X}^{\pi}_T = x+\int_0^T \left(\tilmu(\Theta_t) \pi(\Theta_t) -\frac{\tilsigma^2(\Theta_t)}2\pi^2(\Theta_t)\right)dt
+\int_0^T \tilsigma(\Theta_t) \pi(\Theta_t) dW_t.
\end{equation*}
Now, if $\Theta_t$ follows an autonomous diffusion $d\Theta_t = b(\Theta_t) dt + dW_t$, the above stochastic integral can be replaced by applying It\^o's formula to the function $\Pi(y) = \int_0^y \tilsigma(x)\pi(x) dx$:
\begin{equation*}
\Pi(\Theta_T)-\Pi(\Theta_0) = \int_0^T \left(\tilsigma(\Theta_t)\pi(\Theta_t) b(\Theta_t) +\frac12 (\tilsigma\pi)'(\Theta_t) \right)dt + \int_0^T \tilsigma(\Theta_t)\pi(\Theta_t) dW_t .
\end{equation*}
Indeed, solving the second equation for the stochastic integral, and plugging it into the first equation yields:
\begin{equation*}
\log \tilde{X}^\pi_T = x+\int_0^T \left((\tilmu(\Theta_t)-\tilsigma(\Theta_t) b(\Theta_t)) \pi(\Theta_t) -\frac{\tilsigma^2(\Theta_t)}2\pi^2(\Theta_t)- \frac{(\tilsigma\pi)'(\Theta_t)}{2}\right)dt+( \Pi(\Theta_T)-\Pi(\Theta_0)).
\end{equation*}
This decomposes the logarithmic utility into an integral, which represents the  long-run component, and a residual transitory term, which depends only on the initial and terminal values of the state variable. If the function $\Pi$ is integrable with respect to the invariant measure of $\Theta$, the contribution of the transitory component to the equivalent safe rate $\frac{1}{T} E[\log X^\pi_T]$ is negligible for long horizons.

The myopic probability is key to perform a similar decomposition with power utility. Again, denote by $\tilmu$ the risky asset's drift under the original measure, and by $\hat\mu$ its counterpart under the myopic probability; the corresponding volatility $\tilsigma$ of course has to be the same under both equivalent measures. With logarithmic utility, the optimal portfolio is $\hat\pi_t = \hat \mu_t/\tilsigma_t^2$ even if $\hat \mu_t$ and $\tilsigma_t$ are stochastic \citep{merton.71}. As the definition of the myopic probability requires that the corresponding log-optimal portfolio $\hat\pi_t$ coincides with the optimal portfolio $\tilde\pi_t$ for power utility under the original probability, Girsanov's theorem dictates that the measure change from the original to the myopic probability is governed by the stochastic exponential of $\int_0^T(-\frac{\tilde \mu}{\tilde \sigma} + \tilde \sigma \tilde \pi) dW_t$. This measure change shifts the asset's drift by the same amount, times $\tilde\sigma$, thereby yielding a myopic drift of $\tilde{\sigma}^2 \tilde \pi$, which yields the same optimal policy.
Given this guess for the myopic probability, the finite-horizon bounds follow by routine calculations carried out in the proof of the following lemma:

\begin{lemma}\label{lemfinite}
For a fixed time horizon $T>0$, let $\beta = \frac{\mu^2-\lambda^2}{2\gamma\sigma^2}$ and let the function $w$ be defined as in Lemma~\ref{lem:riccati}. Then, for the the shadow payoff $\tilde X_T$  corresponding to the policy $\tilde \pi(\Upsilon_t) = w(\Upsilon_t)$ and the shadow discount factor $\tilde M_T=e^{-rT}\cale(-\int_0^\cdot \frac\tilmu\tilsigma dW_t)_T$ , the following bounds hold true: 
\begin{align}
E [\tilde X _T^{1-\gamma}] &= \tilde X_0 ^{1-\gamma} e^{(1-\gamma)(r+\beta)T}\hat{E}[e^{(1-\gamma)\left(\tilde q (\Upsilon_0)- \tilde q (\Upsilon_T) \right)}],\label{primbound}\\
E\left[\tilde M _T^{1-\frac{1}{\gamma}}\right]^{\gamma} &= e^{(1-\gamma)(r+\beta)T}\hat{E}\left[e^{(\frac{1}{\gamma}-1)\left(\tilde q (\Upsilon_0)- \tilde q (\Upsilon_T) \right)}\right]^{\gamma}, \label{dualbound}
\end{align}
where $\tilde q (y) := \int_0^y (w(z)-\frac{w'(z)}{1-w(z)}) dz$ and $\hat{E} \left[\cdot\right]$ denotes the expectation with respect to the \emph{myopic probability} $\hat{P}$, defined by
\begin{displaymath}
\frac{d \hat{P}}{dP} = \exp\left(\int_0^T \left(-\frac{\tilde \mu(\Upsilon_t)}{\tilde \sigma(\Upsilon_t)} + \tilde \sigma(\Upsilon_t) \tilde \pi(\Upsilon_t)\right) d W_t - \frac{1}{2}\int_{0}^T\left(-\frac{\tilde \mu(\Upsilon_t)}{\tilde \sigma(\Upsilon_t)}+ \tilde \sigma(\Upsilon_t) \tilde \pi(\Upsilon_t)\right)^2 d t\right).
\end{displaymath}
\end{lemma}

\begin{proof}
First note that $\tilmu, \tilsigma$, and $w$ are functions of $\Upsilon_t$, but the argument is omitted throughout to ease notation. Now, to prove \eqref{primbound}, notice that the frictionless shadow wealth process $\tilde X_t$ with dynamics $\frac{d\tilde{X}_t}{\tilde{X}_t}=w \frac{d\tilde{S}_t}{\tilde{S}_t}+(1-w)\frac{dS^0_t}{S^0_t}$ satisfies:
\begin{equation*}
\tilde X_T^{1-\gamma}=
\tilde{X}_0^{1-\gamma} e^{(1-\gamma)\int_0^T (r+\tilmu w -\frac{\tilsigma^2}{2}w^2) dt
+(1-\gamma)\int_0^T \tilsigma w dW_t}.
\end{equation*}
Hence:
\begin{align*}
\tilde X_T^{1-\gamma} =& \tilde{X}_0^{1-\gamma}\frac{d\hat P}{dP}
e^{\int_0^T ((1-\gamma)(r+\tilmu w -\frac{\tilsigma^2}{2}w^2)
+\frac12(-\frac{\tilmu}{\tilsigma}+\tilsigma w)^2)dt+\int_0^T ((1-\gamma)\tilsigma w-(-\frac{\tilmu}{\tilsigma}+\tilsigma w)) dW_t}.
\end{align*}
Inserting the definitions of $\tilmu$ and $\tilsigma$, the second integrand simplifies to $(1-\gamma)\sigma(\frac{w'}{1-w}-w)$. Similarly, the first integrand reduces to
$(1-\gamma)(r+\frac{\sigma^2}{2}(\frac{w'}{1-w})^2-(1-\gamma)\sigma^2 \frac{w' w}{1-w} +(1-\gamma)\frac{\sigma^2}{2}w^2)$. In summary:
\begin{equation}\label{secrep}
\tilde X_T^{1-\gamma} = \tilde{X}_0^{1-\gamma}\frac{d\hat P}{dP}
e^{(1-\gamma)\int_0^T (r+\frac{\sigma^2}{2}(\frac{w'}{1-w})^2-(1-\gamma)\sigma^2 \frac{w' w}{1-w} +(1-\gamma)\frac{\sigma^2}{2}w^2) dt + (1-\gamma)\int_0^T \sigma(\frac{w'}{1-w}-w) dW_t}.
\end{equation}
The boundary conditions for $w$ and  $w'$ imply $w(0)-\frac{w'(0)}{1-w(0)}=w(\log(u/l))-\frac{w'(\log(u/l))}{1-w(\log(u/l))}=0$; hence, It\^o's formula yields that the local time terms vanish in the dynamics of $\tilde{q}(\Upsilon_t)$:
\begin{equation}\label{eq:subs}
\tilde{q}(\Upsilon_T)-\tilde{q}(\Upsilon_0)=
\int_0^T \left(\mu-\tfrac{\sigma^2}{2}\right) \left(w-\tfrac{w'}{1-w}\right)+\tfrac{\sigma^2}{2}\left(w'-\tfrac{w''(1-w)+w'^2}{(1-w)^2}\right)dt+\int_0^T \sigma \left(w-\tfrac{w'}{1-w}\right)dW_t.  
\end{equation}
Substituting the second derivative $w''$ according to the ODE \eqref{w2ableitung} and using the resulting identity to replace the stochastic integral in \eqref{secrep} yields
\begin{align*}
\tilde X_T^{1-\gamma} =& \tilde{X}_0^{1-\gamma}\frac{d\hat P}{dP}
e^{(1-\gamma)\int_0^T (r+\frac{\sigma^2}{2}w'+(1-\gamma)\frac{\sigma^2}{2}w^2+(\mu-\frac{\sigma^2}{2})w)dt} e^{(1-\gamma)(\tilde{q}(\Upsilon_0)-\tilde{q}(\Upsilon_T))}.
\end{align*}
After inserting the ODE \eqref{eq:wode} for $w$, the first bound thus follows by talking the expectation.

The argument for the second bound is similar. Plugging in the definitions of $\tilmu$ and $\tilsigma$, the shadow discount factor $\tilde{M}_T=e^{-rT}\cale(-\int_0^\cdot \frac\tilmu\tilsigma dW)_T$ and the myopic probability $\hat P$ satisfy:
\begin{align*}
\tilde{M}_T^{1-\frac1\gamma}&=e^{\frac{1-\gamma}{\gamma}\int_0^T \frac{\tilmu}{\tilsigma}dW_t+\frac{1-\gamma}{\gamma}\int_0^T (r+\frac{\tilmu^2}{2\tilsigma^2})dt}\\
&=\frac{d\hat{P}}{dP} e^{\frac{1-\gamma}{\gamma}\int_0^T (\frac{\tilmu}{\tilsigma}-\frac{\gamma}{1-\gamma}(-\frac{\tilmu}{\tilsigma}+\tilsigma w))dW_t+\frac{1-\gamma}{\gamma}\int_0^T(r+\frac{\tilmu^2}{2\tilsigma^2}+\frac{\gamma}{2(1-\gamma)}(-\frac{\tilmu}{\tilsigma}+\tilsigma w)^2)dt}\\
&= \frac{d\hat{P}}{dP} e^{\frac{1-\gamma}{\gamma}\int_0^T \sigma(\frac{w'}{1-w}-w)dW_t+\frac{1-\gamma}{\gamma}\int_0^T(r+\frac{\sigma^2}{2}(\frac{w'}{1-w})^2-(1-\gamma)\sigma^2\frac{w' w}{1-w}+(1-\gamma)\frac{\sigma^2}{2}w^2)dt}.
\end{align*}
Again replace the stochastic integral using \eqref{eq:subs} and the ODE \eqref{w2ableitung}, obtaining
$$ \tilde{M}_T^{1-\frac1\gamma}=\frac{d\hat{P}}{dP} e^{\frac{1-\gamma}{\gamma}\int_0^T (r+\frac{\sigma^2}{2}w'+(1-\gamma)\frac{\sigma^2}{2}w^2+(\mu-\frac{\sigma^2}{2})w)dt}e^{\frac{1-\gamma}{\gamma}(\tilde{q}(\Upsilon_0)-\tilde{q}(\Upsilon_T))}.$$
Inserting the ODE \eqref{eq:wode} for $w$, taking the expectation, and raising it to power $\gamma$, the second bound follows.
\end{proof}

With the finite horizon bounds at hand, it is now straightforward to establish that the policy $\tilde{\pi}(\Upsilon_t)$ is indeed long-run optimal in the frictionless market with price $\tilde{S}_t$.

\begin{lemma}\label{lem:ergodic}
Let $0<\mu/\gamma \sigma^2 \neq 1$ and let $w$ be defined as in Lemma \ref{lem:riccati}.  Then, the risky weight $\tilde{\pi}(\Upsilon_t)=w(\Upsilon_t)$ is long-run optimal with equivalent safe rate $r+\beta$ in the frictionless market with price process $\tilde{S}_t$. The corresponding wealth process (in terms of $\tilde{S}_t$), and the numbers of safe and risky units are given by
\begin{align*}
\tilde{X}_t&=(\xi^0 S^0_0+\xi \tilde{S}_0)\mathcal{E}\left(\int_0^\cdot (r+w(\Upsilon_s) \tilde{\mu}(\Upsilon_s))ds+\int_0^\cdot w(\Upsilon_s)\tilde{\sigma}(\Upsilon_s)dW_s\right)_t, \\
\varphi_{0^-}&=\xi, \quad \varphi_t=w(\Upsilon_t)\tilde{X}_t/\tilde{S}_t \quad \mbox{for }t\geq 0,\\
\varphi^0_{0^-}&=\xi^0, \quad \varphi^0_t=(1-w(\Upsilon_t))\tilde{X}_t/S^0_t \quad \mbox{for }t\geq 0.
\end{align*}
\end{lemma}

\begin{proof}
The formulas for the wealth process and the corresponding numbers of safe and risky units follow directly from the standard frictionless definitions. Now let $\tilde{M}_t$ be the shadow discount factor from Lemma \ref{lemfinite}. Then, standard duality arguments for power utility (cf. Lemma~5 in \citet*{guasoni.robertson.11}) imply that the shadow payoff $\tilde{X}_t^\phi$ corresponding to \emph{any} admissible strategy $\phi_t$  satisfies the inequality
\begin{equation}\label{dualbound1}
\esp{(\tilde X^\phi_T)^{1-\gamma}}^{\frac 1{1-\gamma}}\le 
\esp{\tilde{M}_T^{\frac{\gamma-1}\gamma}}^{\frac\gamma{1-\gamma}} .
\end{equation}
This inequality in turn yields the following upper bound, valid for any admissible strategy $\phi_t$ in the frictionless market with shadow price $\tilde{S}_t$:
\begin{equation}\label{dualbound2}
\liminf_{T \to \infty} \frac{1}{(1-\gamma)T}\log E\left[(\tilde{X}^\phi_T)^{1-\gamma}\right] \le \liminf_{T\rightarrow\infty}
\frac{\gamma}{(1-\gamma)T}\log \esp{\tilde{M}_T^{\frac{\gamma-1}\gamma}}.
\end{equation}
Since the function $\tilde{q}$ is bounded on the compact support of $\Upsilon_t$, the second bound in Lemma \ref{lemfinite} implies that the right-hand side equals $r+\beta$. Likewise, the first bound in the same lemma implies that the shadow payoff $\tilde{X}_t$ (corresponding to the policy $\varphi_t$) attains this upper bound, concluding the proof.
\end{proof}

The next Lemma establishes that the candidate $\tilde{S}_t$ is indeed a shadow price. 

\begin{lemma}\label{lem:strategy}
Let $0<\mu/\gamma \sigma^2 \neq 1$. Then, the number of shares $\varphi_t=w(\Upsilon_t)\tilde{X}_t/\tilde{S}_t$ in the portfolio~$\tilde{\pi}(\Upsilon_t)$ in Lemma~\ref{lem:ergodic} has the dynamics  
\begin{equation}\label{eq:philu}
\frac{d\varphi_t}{\varphi_t}=\left(1-\frac{\mu-\lambda}{\gamma\sigma^2}\right)dL_t-\left(1-\frac{\mu+\lambda}{\gamma\sigma^2}\right)dU_t.
\end{equation}
Thus, $\varphi_t$ increases only when $\Upsilon_t=0$, that is, when $\tilde{S}_t$ equals the ask price, and decreases only when $\Upsilon_t=\log(u/l)$, that is, when $\tilde{S}_t$ equals the bid price.
\end{lemma} 

\begin{proof}
It\^o's formula and the ODE \eqref{w2ableitung} yield
$$dw(\Upsilon_t)=-(1-\gamma)\sigma^2 w'(\Upsilon_t)w(\Upsilon_t)dt+\sigma w'(\Upsilon_t)dW_t+w'(\Upsilon_t)(dL_t-dU_t).$$
Integrating $\varphi_t=w(\Upsilon_t)\tilde{X}_t/\tilde{S}_t$ by parts twice, inserting the dynamics of $w(\Upsilon_t)$, $\tilde{X}_t$, $\tilde{S}_t$, and simplifying yields:
$$\frac{d\varphi_t}{\varphi_t}=\frac{w'(\Upsilon_t)}{w(\Upsilon_t)}d(L_t-U_t).$$
Since $L_t$ and $U_t$ only increase (resp.\ decrease when $\mu/\gamma\sigma^2>1$) on $\{\Upsilon_t=0\}$ and $\{\Upsilon_t=\log(u/l)\}$, respectively, the assertion now follows from the boundary conditions for $w$ and $w'$.
\end{proof}

The optimal growth rate for any frictionless price within the bid-ask spread must be greater or equal than in the original market with bid-ask process $((1-\ve)S_t,S_t)$, because the investor trades at more favorable prices. For a \emph{shadow price}, there is an optimal strategy that only entails buying (resp.\ selling) stocks when $\tilde{S}_t$ coincides with the ask- resp.\ bid price. Hence, this strategy yields the same payoff when executed at  bid-ask prices, and thus is also optimal in the original model with transaction costs. The corresponding equivalent safe rate must also be the same, since the difference due to the liquidation costs vanishes as the horizon grows in~\eqref{eq:longrun}:
 
\begin{proposition}\label{prop:shadow}
For a sufficiently small spread $\ve$, the strategy $(\varphi^0_t,\varphi_t)$ from Lemma \ref{lem:ergodic} is also long-run optimal in the original market with transaction costs, with the same equivalent safe rate.
 \end{proposition}
 
 \begin{proof}
As $\varphi_t$ only increases (resp.\ decreases) when $\tilde{S}_t=S_t$ (resp.\ $\tilde{S}_t=(1-\ve)S_t$), the strategy $(\varphi_t^0,\varphi_t)$ is also self-financing for the bid-ask process $((1-\ve)S_t,S_t)$. Since $S_t \geq \tilde{S}_t \geq (1-\ve)S_t$ and the number $\varphi_t$ of risky shares is always positive, it follows that
\begin{equation}\label{eq:bounds}
\varphi^0_t S^0_t+\varphi_t \tilde{S}_t \geq \varphi^0_t S^0_t +\varphi_t^+(1-\ve)S_t-\varphi^-_t S_t  \geq (1-\tfrac{\ve}{1-\ve} \tilde{\pi}(Y_t))(\varphi^0_t S^0_t +\varphi_t \tilde{S}_t). 
\end{equation}
The shadow risky fraction $\tilde{\pi}(\Upsilon_t)=w(\Upsilon_t)$ is bounded from above by $(\mu+\lambda)/\gamma\sigma^2=\mu/\gamma\sigma^2+O(\ve^{1/3})$. For a sufficiently small spread $\ve$, the strategy $(\varphi_t^0,\varphi_t)$ is therefore also admissible for $((1-\ve)S_t,S_t)$. Moreover, \eqref{eq:bounds} then also yields
\begin{equation}\label{eq:optimalrate}
\begin{split}
\liminf_{T \to \infty} \frac{1}{(1-\gamma)T}\log E\left[(\varphi^0_T S^0_T+\varphi_T^+ (1-\ve)S_T -\varphi_T^- S_T)^{1-\gamma}\right]\\
\qquad = \liminf_{T \to \infty} \frac{1}{(1-\gamma)T}\log E\left[(\varphi^0_T S^0_T+\varphi_T\tilde{S}_T)^{1-\gamma}\right],
\end{split}
\end{equation}
that is, $(\varphi_t^0,\varphi_t)$ has the same growth rate, either with $\tilde{S}_t$ or with $[(1-\ve)S_t,S_t]$. 

For any admissible strategy $(\psi_t^0,\psi_t)$ for the bid-ask spread $[(1-\ve)S_t,S_t]$, set $\tilde{\psi}_t^0=\psi^0_{0^-}-\int_0^{t} \tilde{S}_s/S^0_s d\psi_s$. Then, $(\tilde{\psi}_t^0,\psi_t)$ is a self-financing trading strategy for $\tilde{S}_t$ with $\tilde{\psi}_t^0 \geq \psi_t^0$. Together with $\tilde{S}_t \in [(1-\ve)S_t,S_t]$, the long-run optimality of $(\varphi_t^0,\varphi_t)$ for $\tilde{S}_t$, and~\eqref{eq:optimalrate}, it follows that:
\begin{align*}
&\liminf_{T \to \infty} \frac{1}{T}\frac{1}{(1-\gamma)}\log E\left[(\psi^0_T S^0_T+\psi_T^+ (1-\ve)S_T -\psi_T^- S_T)^{1-\gamma}\right]\\
&\qquad \qquad\leq \liminf_{T \to \infty} \frac{1}{T}\frac{1}{(1-\gamma)}\log E\left[(\tilde{\psi}^0_T S^0_T+\psi_T \tilde{S}_T)^{1-\gamma}\right]\\
&\qquad \qquad\leq \liminf_{T \to \infty} \frac{1}{T}\frac{1}{(1-\gamma)}\log E\left[(\varphi^0_T S^0_T+\varphi_T \tilde{S}_T)^{1-\gamma}\right]\\
&\qquad \qquad= \liminf_{T \to \infty} \frac{1}{T}\frac{1}{(1-\gamma)}\log E\left[(\varphi^0_T S^0_T+\varphi_T^+ (1-\ve)S_T -\varphi_T^- S_T)^{1-\gamma}\right].
\end{align*}
Hence $(\varphi_t^0,\varphi_t)$ is also long-run optimal for $((1-\ve)S_t,S_t)$.
\end{proof}

By putting together the above statements we obtain the following main result:

\begin{theorem}\label{th:opt}
For a small spread $\ve>0$, and $0<\mu/\gamma\sigma^2 \neq 1$, the process $\tilde{S}_t$ in  Lemma~\ref{lem:dynamics} is a shadow price. A long-run optimal policy --- both for the frictionless market with price $\tilde{S}_t$ and in the market with bid-ask prices $(1-\ve)S_t,S_t$ --- is to keep the risky weight $\tilde{\pi}_t$ (in terms of $\tilde{S}_t$) in the no-trade region
$$[\pi_-,\pi_+]=\left[\frac{\mu-\lambda}{\gamma\sigma^2},\frac{\mu+\lambda}{\gamma\sigma^2}\right].$$
As $\ve \downarrow 0$, its boundaries have the asymptotics
\begin{align*}
  \pi_{\pm} = \frac{\mu}{\gamma\sigma^2}
    \pm \left(\frac{3}{4\gamma}\left(\frac{\mu}{\gamma\sigma^2}\right)^2\left(1-\frac{\mu}{\gamma\sigma^2}\right)^2\right)^{1/3} \ve^{1/3} +O(\ve).
\end{align*}
The corresponding equivalent safe rate is:
\[
  r+\beta=r+\frac{\mu^2-\lambda^2}{\gamma\sigma^2}=
    r+\frac{\mu^2}{2\gamma\sigma^2}-\frac{\gamma\sigma^2}{2}\left(\frac{3}{4\gamma} \left(\frac{\mu}{\gamma\sigma^2}\right)^2\left(1-\frac{\mu}{\gamma\sigma^2}\right)^2\right)^{2/3} \ve^{2/3}
    +O(\ve^{4/3}).
\]

If $\mu/\gamma\sigma^2=1$, then $\tilde{S}_t=S_t$ is a shadow price, and it is optimal to invest all wealth in the risky asset at time $t=0$, never to trade afterwards. In this case, the equivalent safe rate is the frictionless value $r+\beta=r+\mu^2/2\gamma\sigma^2$.
\end{theorem}

\begin{proof}
First let $0<\mu/\gamma\sigma^2 \neq 1$. Optimality of the strategy $(\varphi_t^0,\varphi_t)$ associated to $\tilde{\pi}(\Upsilon_t)$ for $\tilde{S}_t$ has been shown in
Lemma~\ref{lem:ergodic}. The asymptotic expansions are an immediate consequence of their counterpart  for~$\lambda$ (cf.\ Lemma~\ref{lem:lambda}) and Taylor expansion. Next, Lemma~\ref{lem:strategy} shows that $\tilde{S}_t$ is a shadow price process in the sense of Definition~\ref{defi:shadow}. Proposition~\ref{prop:shadow} shows that, for small transaction costs $\ve$, the same policy is also optimal, with the same equivalent safe rate, in the original market with bid-ask prices $(1-\ve)S_t,S_t$.

Consider now the degenerate case $\mu/\gamma\sigma^2=1$. Then the optimal strategy in the frictionless model $\tilde{S}_t=S_t$ transfers all wealth to the risky asset at time $t=0$, never to trade afterwards ($\varphi^0_t=0$ and $\varphi_t=\xi+\xi^0 S^0_0/S_0$ for all $t \geq 0$). Hence it is of finite variation and the number of shares never decreases from the unlevered initial position, and increases only at time $t=0$, where the shadow price coincides with the ask price. Thus, $\tilde{S}_t=S_t$ is a shadow price. The remaining assertions then follow as in Proposition~\ref{prop:shadow} above.
\end{proof}

The trading boundaries in this paper are optimal for a long investment horizon, but are also approximately optimal for finite horizons. The following theorem, which complements the main result, makes this point precise:
\begin{theorem}\label{th:finhor}
\label{lem:asympbounds}
Fix a time horizon $T>0$. Then, the finite-horizon equivalent safe rate of the liquidation value $\Xi^\phi_T=\phi^0_T S^0_T+\phi_T^+ (1-\lambda)S_T-\phi_T^- S_T$ associated to \emph{any} strategy $(\phi^0,\phi)$ satisfies the upper bound
\begin{align}\label{eq:upperfinite}
\frac{1}{T}\log E\left[(\Xi^\phi_T)^{1-\gamma}\right]^{\frac{1}{1-\gamma}}&\leq r+\frac{\mu^2-\lambda^2}{2\gamma\sigma^2}+\frac{1}{T}\log(\phi^0_{0^-}+\phi_{0^-}S_0)+\frac{\mu}{\gamma\sigma^2}\frac{\epsilon}{T}+O(\varepsilon^{4/3}) ,\\
\intertext{and the finite-horizon equivalent safe rate of our long-run optimal strategy $(\varphi^0,\varphi)$ satisfies the lower bound}
\label{eq:lowerfinite}
\frac{1}{T}\log E\left[(\Xi^\varphi_T)^{1-\gamma}\right]^{\frac{1}{1-\gamma}}&\geq 
r+\frac{\mu^2-\lambda^2}{2\gamma\sigma^2}+\frac{1}{T}\log(\varphi^0_{0^-}+\varphi_{0^-}S_0)-
\left(\frac{2\mu}{\gamma\sigma^2}+\frac{\varphi_{0^-}S_0}{\varphi^0_{0^-}+\varphi_{0^-}S_0}\right)\frac{\varepsilon}{T}+O(\varepsilon^{4/3}).
\end{align}
In particular, for the same unlevered initial position ($\phi_{0^-}=\varphi_{0^-}\ge 0, \phi^0_{0^-}=\varphi^0_{0^-}\ge 0$), the equivalent safe rates of $(\phi^0,\phi)$ and of the optimal policy $(\varphi^0,\varphi)$ for horizon $T$ differ by at most 
\begin{equation}\label{eq:lossbound}
\frac{1}{T}
\left(\log E\left[(\Xi^\phi_T)^{1-\gamma}\right]^{\frac{1}{1-\gamma}}-
\log E\left[(\Xi^\varphi_T)^{1-\gamma}\right]^{\frac{1}{1-\gamma}}\right)\le
\left(\frac{3\mu}{\gamma\sigma^2}+1\right)\frac{\varepsilon}T+O(\varepsilon^{4/3}).
\end{equation} 
\end{theorem}

This result implies that the horizon, like consumption, only has a second order effect on portfolio choice with transaction costs, because the finite-horizon equivalent safe rate matches, at the leading order $\epsilon^{2/3}$, the equivalent safe rate of the stationary long-run optimal policy, and recovers, in particular, the first-order asymptotics for the finite-horizon value function obtained by \citet[Theorem 4.1]{bichuch.11}. 

\begin{proof}[Proof of Theorem \ref{lem:asympbounds}]
Let $(\phi^0,\phi)$ be any admissible strategy starting from the initial position $(\varphi^0_{0-},\varphi_{0-})$. Then as in the proof of Proposition~\ref{prop:shadow}, we have $\Xi^\phi_T \leq \tilde{X}^\phi_T$ for the corresponding shadow payoff, that is, the terminal value of the wealth process $\tilde{X}^\phi_t=\phi^0_0+\phi_0 \tilde{S}_0+\int_0^t \phi_s d\tilde{S}_s$ corresponding to trading $\phi$ in the frictionless market with price process $\tilde{S}_t$. Hence, Lemma 5 in \citet*{guasoni.robertson.11} and the second bound in Lemma \ref{lemfinite} imply that
\begin{equation}\label{eq:boundup}
\frac{1}{(1-\gamma)T}\log E\left[(\Xi^\phi_T)^{1-\gamma}\right]\leq r+\beta+\frac{1}{T}\log(\varphi^0_{0-}+\varphi_{0-}S_0)
+\frac{\gamma}{(1-\gamma)T}\log\hat{E}\left[e^{(\frac{1}{\gamma}-1)(\tilde{q}(\Upsilon_0)-\tilde{q}(\Upsilon_T))}\right].
\end{equation}
For the strategy $(\varphi^0,\varphi)$ from Lemma~\ref{lem:strategy}, we have $\Xi^\varphi_T \geq (1-\frac{\varepsilon}{1-\varepsilon}\frac{\mu+\lambda}{\gamma\sigma^2})\tilde{X}^\varphi_T$ by the proof of Proposition~\ref{prop:shadow}. Hence the first bound in Lemma \ref{lemfinite} yields
\begin{align}
\frac{1}{(1-\gamma)T}\log E\left[(\Xi^\varphi_T)^{1-\gamma}\right]\geq r+\beta +\frac{1}{T}\log(\varphi^0_{0-}+\varphi_{0-}\tilde{S}_0)
&+\frac{1}{(1-\gamma)T}\log\hat{E}\left[e^{(1-\gamma)(\tilde{q}(\Upsilon_0)-\tilde{q}(\Upsilon_T))}\right] \notag \\
&+\frac{1}{T}\log\left(1-\frac{\varepsilon}{1-\varepsilon}\frac{\mu+\lambda}{\gamma\sigma^2}\right).\label{eq:boundlow}
\end{align}
To determine explicit estimates for these bounds, we first analyze the sign of $\tilde{w}(y)=w-\frac{w'}{1-w}$ and hence the monotonicity of $\tilde{q}(y)=\int_0^y \tilde{w}(z)dz$. Whenever $\tilde{w}=0$, i.e., $w'=w(1-w)$, the derivative of $\tilde{w}$ is 
$$\tilde{w}'=w'-\frac{w''(1-w)+w'^2}{(1-w)^2}=\frac{(1-2\gamma)w' w+\frac{2\mu}{\sigma^2} w'}{1-w}-\left(\frac{w'}{1-w}\right)^2=2\gamma w\left(\frac{\mu}{\gamma\sigma^2}-w\right),$$
where we have used the ODE \eqref{w2ableitung} for the second equality. Since $\tilde{w}$ vanishes at $0$ and $\log(u/l)$ by the boundary conditions for $w$ and $w'$, this shows that the behaviour of $\tilde{w}$ depends on whether the investor's position is leveraged or not. In the absence of leverage, $\mu/\gamma\sigma^2 \in (0,1)$, $\tilde{w}$ is defined on $[0,\log(u/l)]$. It vanishes at the left boundary $0$ and then increases since its derivative is initially positive by the initial condition for $w$. Once the function $w$ has increased to level $\mu/\gamma\sigma^2ß$, the derivative of $\tilde{w}$ starts to become negative; as a result, $\tilde{w}$ begins to decrease until it reaches level zero again at $\log(u/l)$. In particular, $\tilde{w}$ is nonnegative for $\mu/\gamma\sigma^2 \in (0,1)$. 

 In the leverage case $\mu/\gamma\sigma^2>1$, the situation is reversed. Then, $\tilde{w}$ is defined on $[\log(u/l),0]$ and, by the boundary condition for $w$ at $\log(u/l)$, therefore starts to decrease after starting from zero at $\log(u/l)$. Once $w$ has decreased to level $\mu/\gamma\sigma^2$, $\tilde{w}$ starts increasing until it reaches level zero again at $0$. Hence, $\tilde{w}$ is nonpositive for $\mu/\gamma\sigma^2>1$.

Now, consider Case 2 of Lemma \ref{lem:riccati}; the calculations for the other cases follow along the same lines with minor modifications. Then $\mu/\gamma\sigma^2 \in (0,1)$ and $\tilde{q}$ is positive and increasing. Hence, \begin{equation}\label{eq:boundup2}
\frac{\gamma}{(1-\gamma)T}\log\hat{E}\left[e^{(\frac{1}{\gamma}-1)(\tilde{q}(\Upsilon_0)-\tilde{q}(\Upsilon_T))}\right]\leq \frac{1}{T}\int_0^{\log(u/l)}\tilde{w}(y)dy
\end{equation}
and likewise
\begin{equation}\label{eq:boundlow2}
\frac{1}{(1-\gamma)T}\log\hat{E}\left[e^{(1-\gamma)(\tilde{q}(\Upsilon_0)-\tilde{q}(\Upsilon_T))}\right] \geq -\frac{1}{T}\int_0^{\log(u/l)} \tilde{w}(y)dy.
\end{equation}
Since $\tilde{w}(y)=w(y)-w'/(1-w)$, the boundary condions for $w$ imply
\begin{equation}\label{eq:expl3}
\int_0^{\log(u/l)} \tilde{w}(y)dy=\int_0^{\log(u/l)}w(y)dy -\log\left(\frac{\mu-\lambda-\gamma\sigma^2}{\mu+\lambda-\gamma\sigma^2}\right).
\end{equation}
By elementary integration of the explicit formula in Lemma \ref{lem:riccati} and using the boundary conditions from Lemma \ref{lem:smoothpasting} for the evaluation of the result at $0$ resp.\ $\log(u/l)$, the integral of $w$ can also be computed in closed form: 
\begin{align}\label{eq:expl4}
\int_0^{\log(u/l)} w(y)dy
=\tfrac{\frac{\mu}{\sigma^2}-\frac{1}{2}}{\gamma-1}\log\left(\tfrac{1}{1-\varepsilon}\tfrac{(\mu+\lambda)(\mu-\lambda-\gamma\sigma^2)}{(\mu-\lambda)(\mu+\lambda-\gamma\sigma^2)}\right)+\tfrac{1}{2(\gamma-1)}\log\left(\tfrac{(\mu+\lambda)(\mu+\lambda-\gamma\sigma^2)}{(\mu-\lambda)(\mu-\lambda-\gamma\sigma^2)}\right).
\end{align}
As $\epsilon \downarrow 0$, a Taylor expansion and the power series for $\lambda$ then yield
$$\int_0^{\log(u/l)} \tilde{w}(y)dy=\frac{\mu}{\gamma\sigma^2}\varepsilon+O(\varepsilon^{4/3}).$$
Likewise,
$$\log\left(1-\frac{\varepsilon}{1-\varepsilon}\frac{\mu-\lambda}{\gamma\sigma^2}\right)=-\frac{\mu}{\gamma\sigma^2}\varepsilon+O(\varepsilon^{4/3}),$$
as well as
$$
\log(\varphi^0_{0-}+\varphi_{0-}\tilde{S}_0) \geq \log(\varphi^0_{0-}+\varphi_{0-}S_0)- \frac{\varphi_{0-}S_0}{\varphi^0_{0-}+\varphi_{0-}S_0}\varepsilon+O(\varepsilon^2),$$
and the claimed bounds follow from \eqref{eq:boundup} and \eqref{eq:boundup2} resp.\ \eqref{eq:boundlow} and \eqref{eq:boundlow2}.
\end{proof}


\section{Open Problems}

In this section we mention three problems for which, in our view, the above approach holds promise, and the effect of transaction costs is likely to be substantial. Of course, only future research can shed light on this point.

\subsection{Multiple Assets}

In sharp contrast to frictionless models, passing from one to several risky assets is far from trivial with transaction costs. The reason is that, since in the free boundary problem the unknown boundary has one dimension less than the number of risky assets, with one asset it reduces to two points only, but with two assets it already becomes an unknown curve. More importantly, multiple assets introduce novel effects, which defy the one-dimensional intuition, as we now argue.
For example, consider a market with two risky assets with prices $S_t^1$ and $S_t^2$:
\begin{align}
\frac{dS_t^1}{S_t^1} =& \mu_1 dt + \sigma_1 dW^1_t \\
\frac{dS_t^2}{S_t^2} =& \mu_2 dt + \varrho\sigma_2 dW^1_t + \sigma_2 \sqrt{1-\varrho^2} dW^2_t 
\end{align}
where $\mu_1,\sigma_1,\mu_2, \sigma_2>0$, $\varrho \in [-1,1]$, and $W^1, W^2$ are two independent Brownian motions. Even for this simple model with power utility, the solution to the portfolio choice problem is unknown. Some recent papers, e.g., \citet{atkinson2004multi}, \citet{muthuraman2006multidimensional}, \citet{atkinson2006influence}, \citet{law2007correlated}, and \citet{bichuch2011utility}, offer some insights -- and raise a number of questions.

Recall that the frictionless portfolio in the above model is  $\pi = \frac{1}{\gamma} \Sigma^{-1} \mu$, where $\mu=(\mu_1,\mu_2)$ is the vector of excess returns, and $\Sigma$ is the covariance matrix defined as $\Sigma_{11} = \sigma_1^2$, $\Sigma_{12} = \Sigma_{21} = \varrho \sigma_1 \sigma_2$, and $\Sigma_{22} = \sigma_2^2$. In other words:
\[
\pi_1 = \frac{\mu_1-\beta_1 \mu_2}{(1-\varrho^2)\sigma_1^2} 
\qquad
\pi_2 = \frac{\mu_2-\beta_2 \mu_1}{(1-\varrho^2)\sigma_2^2}
\]
where $\beta_i = (\varrho \sigma_1 \sigma_2)/\sigma_i^2$ are the betas of each asset with respect to the other. In particular, for two uncorrelated assets the portfolio separates, in that the optimal weight for each risky asset in the market with all assets equals the optimal weight for the risky asset in a market with that risky asset only. This separation property is intuitive and appealing, and reduces the analysis of frictionless portfolio choice problems with multiple uncorrelated assets to the single asset case. \citet{liu.04} and \citet{guasoni2011long} show that such a separation carries over to transaction cost models with exponential utility.

Surprisingly enough, separation seems to fail with constant relative risk aversion, in that the width of the no-trade region for each asset is affected by the presence of the other, even with zero correlation and logarithmic utility. For example, the heuristics in \citet{law2007correlated} yield the following width for the no-trade region of the first asset, compare their Equation (50): 
\begin{equation}
 H_1= \left(\frac{3\varepsilon}{2 \sigma_1^2}\left[\left(\frac12\mu'\Sigma^{-1}\mu +\sigma_1^2\right)\pi_1^2 -\mu_1 \pi_1^2\right]\right)^{1/3}.
\end{equation}
This quantity clearly depends also on $\mu_2$ and $\sigma_2$ through the total squared Sharpe ratio $\mu'\Sigma^{-1}\mu$, even with zero correlation, and hence differs from the width of the no-trade region with a single risky asset:
\begin{equation}
h_1 = \left(\frac{3\varepsilon}{4} \pi_1^2 (1-\pi_1)^2\right)^{1/3}.
\end{equation}
Further, a simple calculation shows that, if $\varrho=0$, then:
\begin{equation}
H_1^3 - h_1^3 = \frac1{2 \sigma_1^4}\left(\frac{\mu_1 \mu_2}{\sigma_1 \sigma_2}\right)^2.
\end{equation}
In other words, the no-trade region in the larger market is always wider than the no-trade region with one asset, and they coincide only if either asset is useless ($\mu_1=0$ or $\mu_2=0$). In all other cases, the presence of an independent asset increases the no-trade region of the others, presumably because the variation of the position in each asset becomes less important for the overall welfare of the investor than with a single asset.
This observation clearly runs against the common wisdom of fund-separation results for frictionless markets, and has potential implications for intermediation and welfare.

Note that in a frictionless market an investor with power utility is indifferent between trading two uncorrelated assets with Sharpe ratios $\mu_1/\sigma_1$, $\mu_2/\sigma_2$, and a single asset with Sharpe ratio $\sqrt{(\mu_1/\sigma_1)^2+(\mu_2/\sigma_2)^2}$, that is, squared Sharpe ratios and in turn equivalent safe rates add across independent shocks. The above observation suggests that this property  no longer holds with transaction costs, and an important open question is to understand the welfare difference between the two markets.
If the two-asset market is more attractive, then investors benefit from access to individual securities rather than only to a limited number of funds, in contrast to classic fund-separation results. Of course, the question is whether this effect is indeed present and large enough to be relevant.

\subsection{Predictability}

Can future stock returns be predicted with public information? 
And what increase in welfare can one expect from this information? 
Predictably enough, these questions have generated a voluminous literature, which evaluates the statistical significance as well as the in-sample and out-of-sample performance of several predictors that focus either on stock characteristics, such as the dividend-yield and earnings-price ratio, or interest rates, such as the term-spread and the corporate-spread. 

Perhaps less predictably, this voluminous literature remains divided between the weak statistical significance of several models, and the strong economic significance of parameter estimates. On the one hand, the standard errors of the predictability parameters are of the same order of magnitude as the parameter estimates themselves; on the other hand, these estimates -- if valid -- imply a substantial welfare increase. These opposing viewpoints are discussed in \citet{welch2008comprehensive}, who offer a critical view of the empirical literature, and find that most models have poor out-of-sample performance, and \citet{cochrane2008dog}, who argues that the absence of predictability in dividend growth implies the presence of return predictability. 

\citet{kim1996dynamic} introduce a basic model with predictable returns, based on one asset with price $S_t$, and one state variable $\theta_t$:
\begin{align}
dS_t/S_t =& (\mu +\alpha \theta_t) dt + \sigma dW_t,\\
d\theta_t =& -\kappa \theta_t dt + dB_t.
\end{align}
Here $\theta_t$ represents a state variable, like the dividend yield, that helps predict future returns, in that the conditional distribution of $S_T/S_t$ at time $t$ depends on $\theta_t$. The two Brownian motions $W$ and $B$ typically have a substantial negative  correlation $\varrho$. The parameter $\alpha$ controls the predictability of returns, with $\alpha=0$ corresponding to the classical case of IID returns.

In such a market, expected returns change over time, mean-reverting to the average $\mu$. Such variation is detrimental for an investor who adopts the constant policy $\pi = \frac{\mu}{\gamma \sigma^2}$, which is optimal for $\alpha=0$, because time-varying returns increase the dispersion of the final payoff. However, the investor can benefit from \emph{market timing}, that is the ability to adopt an investment policy that depends on the current value of the state variable $\theta_t$. This point is easily seen for logarithmic utility, for which the optimal portfolio is $\pi_t = (\mu +\alpha \theta_t)/\sigma^2$, and the corresponding equivalent safe rate has the simple formula:
\begin{equation}
\lim_{T\rightarrow\infty}\frac1 T\esp{\log X^\pi_T} =
\frac{\mu^2}{2 \sigma^2} + \frac{\alpha^2}{4 \kappa \sigma^2}.
\end{equation}
This expression shows that the investor benefits from stronger signals (larger $\alpha$) and from slower mean reversion of the return rate (smaller $\kappa$), and the same conclusion broadly applies to power utility, even though the formulas become clumsier, as the optimal portfolio includes an intertemporal hedging component that is absent in the logarithmic case. 

The above calculation underlies most estimates of the economic significance of predictability, but obviously ignores transaction costs. This omission may be especially important, as market timing requires active trading, which in turn entails higher costs. In short, while the potential benefit of predicability is clear from the frictionless theory, its potential costs are blissfully ignored, but may be substantial, and a priori may or may not offset benefits.

Remarkably enough, the above model with transaction costs has never been solved, even for logarithmic utility. Intuitively, the solution of this model should lead to a buy curve $\pi_-(\theta)$ and to a sell curve $\pi_+(\theta)$, which describe the no-trade region for each value of the state variable $\theta_t$. Still at an intuitive level, the width of the no-trade region should be wider for values of $\theta$ that are farther from zero, since the portfolio is increasingly likely to return towards the frictionless optimum without trading.

At the technical level, the model includes two state variables: the predictor $\theta_t$, and the current risky weight $\pi_t$. The presence of two state variables in turn implies that the value function satisfies an elliptic linear partial differential equation within the no-trade region, along with the boundary and smooth-pasting conditions at the boundary. The difficulty is to characterize the shape of the no-trade interval $[\pi_-(\theta),\pi_+(\theta)]$, as a function of the state $\theta_t$, along with its implied equivalent safe rate.

Solving such a model can contribute to the predictability debate by clarifying the extent to which the ability to \emph{forecast} future returns can translate into the ability to \emph{deliver} higher returns by trading. When transaction costs are included, it may turn out that potential benefits of market timing are minimal, even if return predictability is statistically significant. 

\subsection{Options Spreads}

Options listed on stock exchanges display much wider bid-ask spreads than their underlying assets. While the spread on a large capitalization stock is typically less than ten basis points, even the most liquid at-the-money options have  spreads of several \emph{percentage} points. To the best of our knowledge, there seems to be no theoretical work that links the bid-ask spread of an asset to the spread of its options.

Of course, in a frictionless, complete market, both spreads are zero, and the option is replicated by a trading strategy in the underlying asset. The problem is that introducing a bid-ask spread for the underlying asset immediately makes the notion of option price ambiguous. Even if the asset follows a geometric Brownian motion, with transaction costs the superreplication price of any call option equals the stock price itself \citep{MR1336872}. Similarly, the subreplication price is zero. Thus, one cannot interpret the bid and ask prices of the option as replication bounds, if the intention is to obtain a realistic spread.

In contrast to the previous two problems, in which the model is clear and the challenges are mathematical, this question poses some conceptual issues at the outset. One possibility is to interpret the bid and ask prices of the options as marginal prices in a partial equilibrium setting. For example, suppose that the bid and ask prices of the asset are exogenous, and follow geometric Brownian motion, with a constant relative bid-ask spread. Suppose also that a representative investor freely trades this asset, and a European option with maturity $T$, as to maximize utility from terminal wealth, either at the same maturity, or at some long horizon.

Since options, unlike stocks, exist in zero net supply, assume that the representative investor's optimal policy is to keep a zero position in the option at all times. In a complete frictionless market, this condition uniquely identifies the option price as the unique arbitrage-free price. With transaction costs, it leaves more flexibility in option price dynamics. Indeed, consider the shadow price corresponding to the utility maximization problem. Since the shadow market is complete, the shadow asset price uniquely identifies a shadow price for the option as the conditional expectation under the risk-neutral probability. For an option of European type with payoff $G(S_T)$, the latter will then be a function $g(t,S_t,Y_t)$ of time, the current stock price, and the current value of the state variable measuring the ratio of risky and safe positions.

Now, suppose that to the original (not shadow) market one adds the option, with a price dynamics equal to the shadow option price, and zero spread. This market is equivalent to the one with the asset only: by contradiction, if some trading strategy delivered a higher utility than the optimum in the asset-only market, the same strategy would also deliver the same or higher utility in the shadow market (by domination), thereby contradicting the definition of a shadow market. Now, the shadow option price depends on the state variable, which is unobservable since market makers cannot see the private positions of market partcipants. However, taking the pointwise maxima $\overline{g}(t,S_t)=\max_{y \in [0,\log(u/l)]} g(t,S_t,y)$ and minima $\underline{g}(t,S_t)=\min_{y \in [0,\log(u/l)]} g(t,S_t,y)$ over \emph{all} values of $Y_t \in [0,\log(u/l)]$, one can obtain observable upper and lower bounds on the option price, which depend on the asset price alone. Such bounds are natural candidates for bid and ask prices of the option, because they are the minimal observable bounds that an option price needs to satisfy if its net demand has to be zero.

The question is whether this construction can predict bid-ask spreads that are consistent with the ones observed in reality, hence much wider than those of the underlying asset.

\bibliographystyle{plainnat}
\bibliography{tractrans}

\begin{thebibliography}{41}
\providecommand{\natexlab}[1]{#1}
\providecommand{\url}[1]{\texttt{#1}}
\expandafter\ifx\csname urlstyle\endcsname\relax
  \providecommand{\doi}[1]{doi: #1}\else
  \providecommand{\doi}{doi: \begingroup \urlstyle{rm}\Url}\fi

\bibitem[Atkinson and Ingpochai(2006)]{atkinson2006influence}
C.~Atkinson and P.~Ingpochai.
\newblock The influence of correlation on multi-asset portfolio optimization
  with transaction costs.
\newblock \emph{J. Comput. Finance}, 10\penalty0 (2):\penalty0 53--96, 2006.

\bibitem[Atkinson and Mokkhavesa(2004)]{atkinson2004multi}
C.~Atkinson and S.~Mokkhavesa.
\newblock Multi-asset portfolio optimization with transaction cost.
\newblock \emph{Appl. Math. Finance}, 11\penalty0 (2):\penalty0 95--123, 2004.

\bibitem[Bichuch(2012)]{bichuch.11}
M.~Bichuch.
\newblock Asymptotic analysis for optimal investment in finite time with
  transaction costs.
\newblock \emph{SIAM J. Financial Math.}, 3\penalty0 (1):\penalty0 433--458,
  2012.

\bibitem[Bichuch and Shreve(2011)]{bichuch2011utility}
M.~Bichuch and S.E. Shreve.
\newblock Utility maximization trading two futures with transaction costs.
\newblock Preprint, 2011.

\bibitem[Carr(1998)]{carr1998randomization}
P.~Carr.
\newblock Randomization and the {A}merican put.
\newblock \emph{Rev. Finan. Stud.}, 11\penalty0 (3):\penalty0 597--626, 1998.

\bibitem[Choi et~al.(2012)Choi, Sirbu, and {\v{Z}}itkovi\'{c}]{choi.al.12}
J.~Choi, M.~Sirbu, and {\v{Z}}itkovi\'{c}.
\newblock Shadow prices and well-posedness in the problem of optimal investment
  and consumption with transaction costs.
\newblock Preprint, 2012.

\bibitem[Cochrane(2008)]{cochrane2008dog}
J.H. Cochrane.
\newblock The dog that did not bark: A defense of return predictability.
\newblock \emph{Rev. Finan. Stud.}, 21\penalty0 (4):\penalty0 1533--1575, 2008.

\bibitem[Constantinides(1986)]{constantinides.86}
G.M. Constantinides.
\newblock {Capital market equilibrium with transaction costs}.
\newblock \emph{J. Polit. Economy}, 94\penalty0 (4):\penalty0 842--862, 1986.

\bibitem[Cvitani{\'c} and Karatzas(1996)]{MR1384221}
J.~Cvitani{\'c} and I.~Karatzas.
\newblock Hedging and portfolio optimization under transaction costs: a
  martingale approach.
\newblock \emph{Math. Finance}, 6\penalty0 (2):\penalty0 133--165, 1996.

\bibitem[Davis and Norman(1990)]{MR1080472}
M.~H.~A. Davis and A.~R. Norman.
\newblock Portfolio selection with transaction costs.
\newblock \emph{Math. Oper. Res.}, 15\penalty0 (4):\penalty0 676--713, 1990.

\bibitem[Dumas(1991)]{dumas.91}
B.~Dumas.
\newblock Super contact and related optimality conditions.
\newblock \emph{J. Econom. Dynam. Control}, 15\penalty0 (4):\penalty0 675--685,
  1991.

\bibitem[Dumas and Luciano(1991)]{dumas.luciano.91}
B.~Dumas and E.~Luciano.
\newblock {An exact solution to a dynamic portfolio choice problem under
  transactions costs}.
\newblock \emph{J. Finance}, 46\penalty0 (2):\penalty0 577--595, 1991.

\bibitem[Gerhold et~al.(2011{\natexlab{a}})Gerhold, Guasoni, Muhle-Karbe, and
  Schachermayer]{gerhold.al.11}
S.~Gerhold, P.~Guasoni, J.~Muhle-Karbe, and W.~Schachermayer.
\newblock {Transaction costs, trading volume, and the liquidity premium}.
\newblock Preprint, 2011{\natexlab{a}}.

\bibitem[Gerhold et~al.(2011{\natexlab{b}})Gerhold, Muhle-Karbe, and
  Schachermayer]{gerhold2010dual}
S.~Gerhold, J.~Muhle-Karbe, and W.~Schachermayer.
\newblock {The dual optimizer for the growth-optimal portfolio under
  transaction costs}.
\newblock \emph{Finance Stoch.}, \penalty0 (To appear), 2011{\natexlab{b}}.

\bibitem[Guasoni and Muhle-Karbe(2011)]{guasoni2011long}
P.~Guasoni and J.~Muhle-Karbe.
\newblock Long horizons, high risk-aversion, and endogenous spreads.
\newblock Preprint, 2011.

\bibitem[Guasoni and Robertson(2012)]{guasoni.robertson.11}
P.~Guasoni and S.~Robertson.
\newblock Portfolios and risk premia for the long run.
\newblock \emph{Ann. Appl. Probab.}, 22\penalty0 (1):\penalty0 239--284, 2012.

\bibitem[Guasoni et~al.(2008)Guasoni, R{\'a}sonyi, and
  Schachermayer]{MR2398764}
P.~Guasoni, M.~R{\'a}sonyi, and W.~Schachermayer.
\newblock Consistent price systems and face-lifting pricing under transaction
  costs.
\newblock \emph{Ann. Appl. Probab.}, 18\penalty0 (2):\penalty0 491--520, 2008.

\bibitem[Herzegh and Prokaj(2011)]{herzegh.prokaj.11}
A.~Herzegh and V.~Prokaj.
\newblock {Shadow price in the power utility case}.
\newblock Preprint, 2011.

\bibitem[Jane{\v{c}}ek and Shreve(2004)]{MR2048827}
K.~Jane{\v{c}}ek and S.~E. Shreve.
\newblock Asymptotic analysis for optimal investment and consumption with
  transaction costs.
\newblock \emph{Finance Stoch.}, 8\penalty0 (2):\penalty0 181--206, 2004.

\bibitem[Jouini and Kallal(1995)]{MR1338025}
E.~Jouini and H.~Kallal.
\newblock Martingales and arbitrage in securities markets with transaction
  costs.
\newblock \emph{J. Econom. Theory}, 66\penalty0 (1):\penalty0 178--197, 1995.

\bibitem[Kallsen and Muhle-Karbe(2010)]{MR2676941}
J.~Kallsen and J.~Muhle-Karbe.
\newblock On using shadow prices in portfolio optimization with transaction
  costs.
\newblock \emph{Ann. Appl. Probab.}, 20\penalty0 (4):\penalty0 1341--1358,
  2010.

\bibitem[Kim and Omberg(1996)]{kim1996dynamic}
T.S. Kim and E.~Omberg.
\newblock Dynamic nonmyopic portfolio behavior.
\newblock \emph{Rev. Finan. Stud.}, 9\penalty0 (1):\penalty0 141--161, 1996.

\bibitem[Kramkov and Schachermayer(1999)]{kramkov1999asymptotic}
D.~Kramkov and W.~Schachermayer.
\newblock The asymptotic elasticity of utility functions and optimal investment
  in incomplete markets.
\newblock \emph{Ann. Appl. Probab.}, 9\penalty0 (3):\penalty0 904--950, 1999.

\bibitem[Law et~al.(2007)Law, Lee, Howison, and Dewynne]{law2007correlated}
S.L. Law, C.F. Lee, S.~Howison, and J.N. Dewynne.
\newblock Correlated multi-asset portfolio optimisation with transaction cost.
\newblock Preprint, 2007.

\bibitem[Liu(2004)]{liu.04}
H.~Liu.
\newblock Optimal consumption and investment with transaction costs and
  multiple risky assets.
\newblock \emph{J. Finance}, 59\penalty0 (1):\penalty0 289--338, 2004.

\bibitem[Liu and Loewenstein(2002)]{liu2002optimal}
H.~Liu and M.~Loewenstein.
\newblock {Optimal portfolio selection with transaction costs and finite
  horizons}.
\newblock \emph{Rev. Finan. Stud.}, 15\penalty0 (3):\penalty0 805--835, 2002.

\bibitem[Liu(2007)]{liu.07}
J.~Liu.
\newblock Portfolio selection in stochastic enviroments.
\newblock \emph{Rev. Finan. Stud.}, 20\penalty0 (1):\penalty0 1--39, 2007.

\bibitem[Liu and Muhle-Karbe(2012)]{muhlekarbe.liu.12}
R.~Liu and J.~Muhle-Karbe.
\newblock {Portfolio selection with small transaction costs and binding
  portfolio constraints}.
\newblock Preprint, 2012.

\bibitem[Loewenstein(2000)]{loewenstein.00}
M.~Loewenstein.
\newblock On optimal portfolio trading strategies for an investor facing
  transactions costs in a continuous trading market.
\newblock \emph{J. Math. Econom.}, 33\penalty0 (2):\penalty0 209--228, 2000.

\bibitem[Magill and Constantinides(1976)]{MR0469196}
M.~J.~P. Magill and G.~M. Constantinides.
\newblock Portfolio selection with transactions costs.
\newblock \emph{J. Econom. Theory}, 13\penalty0 (2):\penalty0 245--263, 1976.

\bibitem[Merton(1969)]{merton.69}
R.C. Merton.
\newblock {Lifetime portfolio selection under uncertainty: The continuous-time
  case}.
\newblock \emph{Rev. Econ. Statist.}, 51\penalty0 (3):\penalty0 247--257, 1969.

\bibitem[Merton(1971)]{merton.71}
Robert~C. Merton.
\newblock Optimum consumption and portfolio rules in a continuous-time model.
\newblock \emph{J. Econom. Theory}, 3\penalty0 (4):\penalty0 373--413, 1971.

\bibitem[Muthuraman and Kumar(2006)]{muthuraman2006multidimensional}
K.~Muthuraman and S.~Kumar.
\newblock Multidimensional portfolio optimization with proportional transaction
  costs.
\newblock \emph{Math. Finance}, 16\penalty0 (2):\penalty0 301--335, 2006.

\bibitem[Rogers(2004)]{MR2076549}
L.~C.~G. Rogers.
\newblock Why is the effect of proportional transaction costs
  {$O(\delta^{2/3})$}?
\newblock In \emph{Mathematics of Finance}, volume 351 of \emph{Contemp.
  Math.}, pages 303--308. Amer. Math. Soc., Providence, RI, 2004.

\bibitem[Shreve and Soner(1994)]{MR1284980}
S.~E. Shreve and H.~M. Soner.
\newblock Optimal investment and consumption with transaction costs.
\newblock \emph{Ann. Appl. Probab.}, 4\penalty0 (3):\penalty0 609--692, 1994.

\bibitem[Skorohod(1962)]{skorokhod.62}
A.~V. Skorohod.
\newblock Stochastic equations for diffusion processes with boundaries. {II}.
\newblock \emph{Teor. Verojatnost. i Primenen.}, 7:\penalty0 5--25, 1962.

\bibitem[Soner and Touzi(2012)]{soner.touzi.12}
H.~M. Soner and N.. Touzi.
\newblock {Homogenization and asymptotics for small transaction costs}.
\newblock Preprint, 2012.

\bibitem[Soner et~al.(1995)Soner, Shreve, and Cvitani{\'c}]{MR1336872}
H.~M. Soner, S.~E. Shreve, and J.~Cvitani{\'c}.
\newblock There is no nontrivial hedging portfolio for option pricing with
  transaction costs.
\newblock \emph{Ann. Appl. Probab.}, 5\penalty0 (2):\penalty0 327--355, 1995.

\bibitem[Taksar et~al.(1988)Taksar, Klass, and Assaf]{MR942619}
M.~Taksar, M.~J. Klass, and D.~Assaf.
\newblock A diffusion model for optimal portfolio selection in the presence of
  brokerage fees.
\newblock \emph{Math. Oper. Res.}, 13\penalty0 (2):\penalty0 277--294, 1988.

\bibitem[Welch and Goyal(2008)]{welch2008comprehensive}
I.~Welch and A.~Goyal.
\newblock A comprehensive look at the empirical performance of equity premium
  prediction.
\newblock \emph{Rev. Finan. Stud.}, 21\penalty0 (4):\penalty0 1455--1508, 2008.

\bibitem[Whalley and Wilmott(1997)]{whalley.wilmott.97}
A.~E. Whalley and P.~Wilmott.
\newblock An asymptotic analysis of an optimal hedging model for option pricing
  with transaction costs.
\newblock \emph{Math. Finance}, 7\penalty0 (3):\penalty0 307--324, 1997.

\end{thebibliography}

\end{document}